\documentclass[draft,onecolumn]{IEEEtran}
\IEEEoverridecommandlockouts

\usepackage{mathtools,amssymb,lipsum, nccmath}

\usepackage{cuted}
\usepackage{lipsum}
\usepackage{mathtools}
\usepackage{cuted}

\usepackage{verbatim}
\usepackage{cite}
\usepackage{amsmath,amssymb,amsfonts,amsthm,stmaryrd}
\usepackage{algorithmic}
\usepackage{graphicx}
\usepackage{tikz}
\usepackage{graphicx}
\usepackage{bbm}
\usepackage{bm}
\usepackage{textcomp}
\usepackage{xcolor}
\usepackage{graphicx}
\usepackage{pgfplots}
\usepackage{cite}
\usepackage{systeme}

\usetikzlibrary{patterns}
\usepackage{bigints}
\usepackage{color}
\usepgfplotslibrary{fillbetween}

\usepackage{soul}
\usepackage{bm}
\usepackage{textcomp}
\usepackage{xcolor}
\usepackage{subcaption} 
\DeclareMathOperator*{\argmax}{arg\,max}
\DeclareMathOperator*{\argmin}{arg\,min}
\def\BibTeX{{\rm B\kern-.05em{\sc i\kern-.025em b}\kern-.08em
    T\kern-.1667em\lower.7ex\hbox{E}\kern-.125emX}}

\newtheorem{theorem}{Theorem}
\newtheorem{remark}{Remark}
\newtheorem{definition}{Definition}

\newtheorem{lemma}{Lemma}

\def\proof{\noindent{\it Proof}. \ignorespaces}
\def\endproof{\vbox{\hrule height0.6pt\hbox{\vrule height1.3ex%
width0.6pt\hskip0.8ex\vrule width0.6pt}\hrule height0.6pt}}

\newcommand{\N}{\mathbb{N}}

\newcommand{\E}{\mathbb{E}}

\usepackage[utf8]{inputenc} 
\usepackage[T1]{fontenc}
\usepackage{url}
\usepackage{ifthen}
\usepackage{cite}

\begin{document}

\title{Strategic Successive Refinement Coding for Bayesian Persuasion with Two Decoders}

 
 \author{\IEEEauthorblockN{Rony Bou Rouphael and Ma\"{e}l Le Treust
\thanks{
Ma\"el Le Treust gratefully acknowledges financial support from INS2I CNRS, DIM-RFSI, SRV ENSEA, UFR-ST UCP, INEX Paris Seine Initiative and IEA Cergy-Pontoise. This research has been conducted as part of the project Labex MME-DII (ANR11-LBX-0023-01).} 
}\\
\IEEEauthorblockA{
ETIS UMR 8051, CY Cergy-Paris Université, ENSEA, CNRS,\\
6, avenue du Ponceau, 95014 Cergy-Pontoise CEDEX, FRANCE\\
Email: \{rony.bou-rouphael ; mael.le-treust\}@ensea.fr}\\
 }

\maketitle

\begin{abstract}
We study the multi-user Bayesian persuasion game between one encoder and two decoders, where the first decoder is  better informed than the second decoder. We consider two perfect links, one to the first decoder only, and the other to both decoders. We consider that the encoder and both decoders are endowed with distinct and arbitrary distortion functions. We investigate the strategic source coding problem in which the encoder commits to an encoding while the decoders select the sequences of symbols that minimize their long-run respective distortion functions. We characterize the optimal encoder distortion value by considering successive refinement coding with respect to a specific probability distribution which involves two auxiliary random variables, and captures the incentive constraints of both decoders.
\end{abstract}



\section{Introduction} 

The optimization of distinct and arbitrary distortion functions resulting from the communication between several autonomous devices with non-aligned objectives is under study. This problem was originally formulated in the game theory literature and referred to as the sender-receiver game, where the amount of information transmitted is generally unrestricted. In the seminal paper \cite{crawford1982}, Crawford and Sobel investigate the Nash equilibrium solution of the cheap talk game in which the encoder and the decoder have distinct objectives and choose their coding strategies simultaneously. In \cite{KamenicaGentzkow11}, Kamenica and Gentzkow formulate the Bayesian persuasion game
in which the encoder is the Stackelberg leader and the decoder is the Stackelberg follower. More recently, Koessler et al. in \cite{KoesslerLaclauTomala21} investigate games of information design where multiple encoders influence the behavior of several decoders.
As a motivating example, one could think of a company trying to convince investors into putting money on a certain number of projects, or a job seeker trying to persuade recruiters to be hired.

This problem is an attractive multi-disciplinary subject of study. The Nash equilibrium solution is investigated for multi-dimensional sources and quadratic distortion functions in \cite{sar1}, \cite{SaritasFurrerGeziciLinderYukselISIT2019}, whereas the Stackelberg solution is studied in \cite{sar2}. The computational aspects of the persuasion game are considered in \cite{dughmi}. The strategic communication problem with a noisy channel is investigated in \cite{AkyolLangbortBasar15}, \cite{akyol2017information}, \cite{LeTreustTomala(Allerton)16}, \cite{jet}, and four different scenarios of strategic communication are studied in \cite{pointtopoint}. The case where the decoder privately observes a signal correlated to the state, also referred to as the Wyner-Ziv setting \cite{wyner-it-1976}, is studied in \cite{akyol2016role}, \cite{corsica2020} and \cite{LeTreustTomalaISIT21}. Vora and Kulkarni investigate the achievable rates for the strategic communication problem in \cite{voraachievable}, \cite{voraextraction} where the decoder is the Stackelberg leader.

In this paper, we investigate a Bayesian persuasion game with two decoders and restricted communication. We consider an i.i.d. source of information and we suppose that 
the observation of the first decoder contains the observation of the second decoder, as in Fig.~\ref{fig:successive00}. 
More specifically, we assume that the encoder $\mathcal{E}$ selects and announces beforehand the compression scheme to be implemented. Upon receipt of the indices, the decoders $\mathcal{D}_1$ and $\mathcal{D}_2$ update their Bayesian beliefs over the source sequence and select the action sequences that minimizes their respective distortion functions. We characterize the optimal encoder distortion value obtained via the successive refinement coding with respect to the distribution that involves two auxiliary random variables, and that satisfies both decoders incentive constraints.

\begin{figure}
    \centering
    \begin{tikzpicture}[xscale=4,yscale=1.3]
\draw[thick,->](0.5,0.25)--(0.75,0.25);
\draw[thick,-](0.75,-0.3)--(1.25,-0.3)--(1.25,0.8)--(0.75,0.8)--(0.75,-0.3);
\node[below, black] at (1,0.48) {$\mathcal{E}$};
\node[above, black] at (1,1.15) {$d_e^n(\sigma,\tau_1,\tau_2)$};
\node[above, black] at (2.4,1.15) {$d_1^n(\sigma,\tau_1)$};
\node[above, black] at (2.4,-1.3) {$d_2^n(\sigma,\tau_2)$};
\draw[thick, ->](1.25,0.7)--(2.25,0.7);
\draw[thick, ->](2,-0.2)--(2,0.6)--(2.25,0.6);
\draw[thick, ->](1.25,-0.2)--(2.25,-0.2);
\draw[thick,-](2.25,0.5)--(2.55,0.5)--(2.55,1)--(2.25,1)--(2.25,0.5);
\node[below, black] at (2.4,1) { $\mathcal{D}_{1}$};
\draw[thick,-](2.25,-0.5)--(2.55,-0.5)--(2.55,0)--(2.25,0)--(2.25,-0.5);
\node[below, black] at (2.4,0) { $\mathcal{D}_{2}$};
\draw[thick,->](2.55,0.75)--(2.75,0.75);
\draw[thick,->](2.55,-0.25)--(2.75,-0.25);
\node[right, black] at (2.75,-0.25) {$V^n_2$};
\node[right, black] at (2.75,0.75) {$V^n_1$};
\node[above, black] at (0.62,0.25) {$U^n$};
\node[above, black] at (1.78,0.66){$M_1\in \{1,..2^{\lfloor nR_1 \rfloor}\}$};
\node[below, black] at (1.78,-0.15){$M_2\in \{1,..2^{\lfloor nR_2 \rfloor}\}$};
  
    \end{tikzpicture}
    \caption{Strategic successive refinement source coding. 
    }
    \label{fig:successive00}
\end{figure}
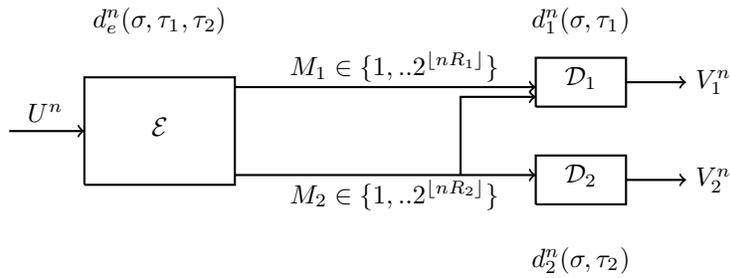

\subsection{Notations}


Let $\mathcal{E}$ denote the encoder and $\mathcal{D}_i$ denote the decoder $i \in \{1,2\}$. Notations $U^n$ and $V^n_i$ denote the $n$-sequences of random variables of source information $u^n=(u_{1},...,u_{n}) \in \mathcal{U}^n $, and decoder $\mathcal{D}_i$ actions $v^n_i \in \mathcal{V}^n_i$ respectively for $i \in \{1,2\}$. Calligraphic fonts $\mathcal{U}$ 
 and $\mathcal{V}_i$ denote the alphabets and lowercase letters $u$ 
 and $v_i$ denote the realizations.
For a discrete random variable $X,$ we denote by $\Delta(\mathcal{X})$ the probability simplex, i.e. the set of probability distributions over $\mathcal{X},$ and by $\mathcal{P}_{X}(x)$ the probability mass function $\mathbb{P}\{X=x\}$. Notation $X -\!\!\!\!\minuso\!\!\!\!- Y -\!\!\!\!\minuso\!\!\!\!- Z$ stands for the Markov chain property $\mathcal{P}_{Z|XY}=\mathcal{P}_{Z|Y}$. 

 \section{System Model}
In this section, we aim at formulating the coding problem. We assume that the information source $U$ follows the independent and identically distributed (i.i.d) probability distribution $\mathcal{P}_{U}\in\Delta(\mathcal{U})$.
\begin{definition}
Let $R_1,R_2 \in \mathbb{R}_{+}^2=[0,+\infty[^2$, where $[0,+\infty[$ denotes the set of non-negative real numbers, and $n\in\mathbb{N}^{\star}=\mathbb{N}\backslash\{0\}$. The encoding $\sigma$ and decoding $\tau_i$ strategies of the encoder $\mathcal{E}$ and decoders $\mathcal{D}_i$, $i \in \{1,2\}$ are defined by
\begin{align}
    \sigma:& U^n \longrightarrow \Delta(\{1,2,..2^{\lfloor nR_1 \rfloor}\}\times\{1,2,..2^{\lfloor nR_2\rfloor}\}),\\
\tau_1:& \{1,2,..2^{\lfloor nR_1 \rfloor}\}\times\{1,2,..2^{\lfloor nR_2\rfloor}\} 
\longrightarrow \Delta(\mathcal{V}_1^n),\\
\tau_2:& \{1,2,..2^{\lfloor nR_2\rfloor}\} 
\longrightarrow \Delta(\mathcal{V}_2^n),
 \end{align} 
where $\lfloor x \rfloor=\max\{m\in\mathbb{Z}, \  m \leq x\}$ for $x\in\mathbb{R}$. We denote by $\mathcal{S}(n,R_1,R_2)$ the set of coding triplets ($\sigma,\tau_1,\tau_2)$.
\end{definition}
The stochastic coding strategies ($\sigma,\tau_1,\tau_2) \in \mathcal{S}(n,R_1,R_2)$ induce a joint probability distribution $\mathcal{P}^{\sigma,\tau_1,\tau_2} \in \Delta(U^n \times  \{1,2,..2^{\lfloor nR_1 \rfloor}\}\times\{1,2,..2^{\lfloor nR_2\rfloor}\} \times V^n_1\times V^n_2)$ defined by
\begin{align}
&\forall (u^n,m_1,m_2,v_1^n,v_2^n), \quad \mathcal{P}^{\sigma,\tau_1,\tau_2}(u^n,m_1,m_2,v_1^n,v_2^n) =
&\bigg(\prod_{t=1}^n\mathcal{P}_{U}(u_t)\bigg)\sigma(m_1,m_2|u^n)\tau_1(v_1^n|m_1,m_2)\tau_2(v_2^n|m_2). \label{jointproba}
\end{align}

\begin{definition} \label{singlelongrun} 
We consider arbitrary single-letter distortion functions $d_e:\mathcal{U}\times\mathcal{V}_1\times\mathcal{V}_2 \longrightarrow \mathbb{R}$ for the encoder $\mathcal{E}$, $d_1:\mathcal{U}\times\mathcal{V}_1 \longrightarrow \mathbb{R}$ for the decoder $\mathcal{D}_1$ and $d_2:\mathcal{U}\times\mathcal{V}_2 \longrightarrow \mathbb{R}$ for the decoder $\mathcal{D}_2$.
The long-run distortion functions are defined by  
\begin{align*}
&d_e^n(\sigma,\tau_1,\tau_2)= \mathbb{E}_{\sigma,\tau_1,\tau_2}\Bigg[\frac{1}{n}\sum_{t=1}^n d_e(U_t,V_{1,t},V_{2,t})\Bigg]  \\
&=\sum_{u^n, v_1^n, v_2^n} \mathcal{P}_{U^nV_1^nV_2^n}^{\sigma,\tau_1,\tau_2}(u^n, v_1^n, v_2^n)\cdot\Bigg[\frac{1}{n}\sum_{t=1}^n d_e(u_t,v_{1,t},v_{2,t})\Bigg], \\
&d^n_{1}(\sigma,\tau_1)= 
\sum_{u^n, v_1^n} \mathcal{P}_{U^nV_1^n}^{\sigma,\tau_1}(u^n, v_1^n)\cdot\Bigg[\frac{1}{n}\sum_{t=1}^n d_1(u_t,v_{1,t})\Bigg], \nonumber\\ 
&d^n_{2}(\sigma,\tau_2)=\sum_{u^n, v_2^n} \mathcal{P}_{U^nV_2^n}^{\sigma,\tau_2}(u^n, v_2^n)\cdot\Bigg[\frac{1}{n}\sum_{t=1}^n d_2(u_t,v_{2,t})\Bigg].
\label{dn2} 
\end{align*}
\end{definition}

 In the above equations, $\mathcal{P}^{\sigma,\tau_1,\tau_2}_{U^nV_1^nV_2^n}$, $\mathcal{P}^{\sigma,\tau_1}_{U^n V_1^n}$ and $\mathcal{P}^{\sigma,\tau_2}_{U^n V_2^n}$ denote the marginal distributions of $\mathcal{P}^{\sigma,\tau_1,\tau_2}$ defined in \eqref{jointproba} over 
 $(U^n,V_1^n,V_2^n)$, $(U^n,V_1^n)$, and $(U^n,V_2^n)$ respectively.  
\begin{definition} For any encoding strategy $\sigma,$ the set of best-response strategies of decoder $i \in \{1,2\}$ is defined by
\begin{equation}
    BR_{i}(\sigma)=\Big\{\tau_i, {d_i}^n(\sigma,\tau_i) \leq {d_i}^n(\sigma,\Tilde{\tau_i}), \forall \ \Tilde{\tau_i} \Big\}. 
\end{equation}
\end{definition}

If several pairs of best-response strategies $(\tau_1,\tau_2)\in BR_{1}(\sigma)\times BR_{2}(\sigma)$ are available, we assume that the worst pair $(\tau_1,\tau_2)$, from the encoder perspective, is selected. Therefore, the solution is robust to the exact specification of the decoding strategies. For $(R_1,R_2)\in \mathbb{R}_{+}^{2}$ and $n \in \mathbb{N}^{\star}$, the coding problem under study is 
\begin{equation}
D_e^n(R_1,R_2)=\underset{\sigma}{\inf}\underset{\tau_1 \in BR_{1}(\sigma), \atop \tau_2 \in BR_{2}(\sigma) }{\max} d_e^n(\sigma,\tau_1,\tau_2). \label{LP}
\end{equation}

\begin{remark}
Suppose that the decoders choose, among their best-response strategies, the pair that also minimizes the encoder distortion. This ``optimistic'' coding problem writes
\begin{equation}
D_o^n(R_1,R_2)=\underset{\sigma}{\min}\underset{\tau_1 \in BR_{1}(\sigma), \atop \tau_2 \in BR_{2}(\sigma) }{\min} d_e^n(\sigma,\tau_1,\tau_2). 
\end{equation}
For generic problems $D_o^n(R_1,R_2)=D_e^n(R_1,R_2)$ \cite[pp. 8]{jet}.
 \end{remark}
The operational significance of 
\eqref{LP} corresponds to the persuasion game that is played in the following steps:   
\begin{itemize}
\item Encoder $\mathcal{E}$ chooses, announces the encoding $\sigma$.
\item Sequence $U^n$ is drawn i.i.d with distribution $\mathcal{P}_U$.
\item Messages $(M_1,M_2)$ are encoded according to $\mathcal{P}^{\sigma}_{M_1M_2|U^n}$.
\item Knowing $\sigma$, 
decoder $\mathcal{D}_1$ observes $(M_1,M_2)$ and draws 
$V_1^n$ according to 
$\tau_1\in BR_{1}(\sigma)$, and 
decoder $\mathcal{D}_2$ observes  $M_2$ and draws 
$V_2^n$ according to 
$\tau_2\in BR_{2}(\sigma)  $.
    \item Distortion values are $d_e^n(\sigma,\tau_1,\tau_2)$, $d_1^n(\sigma,\tau_1)$, $d_2^n(\sigma,\tau_2) $.
\end{itemize}

\begin{lemma}\label{lemma:subadditive}
The sequence $\big(n D_e^n(R_1,R_2)\big)_{n\in \N^{\star}}$ is sub-additive. 
\end{lemma}
The proof is stated in Appendix \ref{Appendix A}.

\section{Main Result}

In this section, we characterize the asymptotic behaviour of $D_e^n(R_1,R_2)$. Our solution combines the decoders incentive constraints with the information constraints of the successive refinement source coding.

  
\begin{definition}\label{def:characterization}
We consider two auxiliary random variables  $W_1\in\mathcal{W}_1$ and $W_2\in\mathcal{W}_2$ with  $|\mathcal{W}_i|=|\mathcal{V}_i|$, for $i \in \{1,2\}$.  For $(R_1,R_2) \in \mathbb{R}^{2}_{+}$, we define 
\begin{align}
\mathbb{Q}_0(R_1,R_2) = \big\{& \mathcal{Q}_{W_1W_2|U},\quad   R_2 \geq I(U;W_2), \nonumber \\ &\quad R_1+R_2 \geq I(U;W_1,W_2) \big\}.\label{srqq0}
\end{align}
For every distribution $\mathcal{Q}_{W_1W_2|U}\in\Delta(\mathcal{W}_1\times \mathcal{W}_2)^{|\mathcal{U}|}$, we define
\begin{align}
  \mathbb{Q}_{1}(\mathcal{Q}_{W_1W_2|U})=&\argmin_{\mathcal{Q}_{V_1| W_1W_2}}\mathbb{E}_{\mathcal{Q}_{W_1W_2|U}\atop \mathcal{Q}_{V_1|W_1W_2}}\Big[d_{1}(U ,V_1)\Big],\\
\mathbb{Q}_{2}(\mathcal{Q}_{W_2|U})=&\argmin_{\mathcal{Q}_{V_2| W_2}}\mathbb{E}_{\mathcal{Q}_{W_2|U} \atop\mathcal{Q}_{V_2| W_2}}\Big[d_{2}(U,V_2)\Big].
 \end{align}
Note that $\mathbb{Q}_{1}(\mathcal{Q}_{W_1W_2|U}) \in \Delta(\mathcal{V}_1)^{|\mathcal{W}_1\times\mathcal{W}_2|}$ and $\mathbb{Q}_{2}(\mathcal{Q}_{W_2|U}) \in \Delta(\mathcal{V}_2)^{|\mathcal{W}_2|}$. The encoder's optimal distortion is defined by
 \begin{align}
&D^{\star}_e(R_1,R_2) \nonumber\\
=&\underset{\mathcal{Q}_{W_1W_2|U}\atop\in\mathbb{Q}_0(R_1,R_2)}{\inf}\underset{\mathcal{Q}_{V_1|W_1W_2}\in \mathbb{Q}_{1}(\mathcal{Q}_{W_1W_2|U}) \atop \mathcal{Q}_{V_2|W_2}\in \mathbb{Q}_{2}(\mathcal{Q}_{W_2|U})}{\max}\mathbb{E} \Big[    d_e(U,V_1,V_2)\Big], \label{optdistoooo1}
\end{align} 
where the expectation in \eqref{optdistoooo1} is evaluated with respect to $\mathcal{P}_{U} \mathcal{Q}_{W_1W_2|U}\mathcal{Q}_{V_1|W_1W_2}\mathcal{Q}_{V_2|W_2}$. 
\end{definition}

\begin{remark}
The random variables $U,W_1,W_2,V_1,V_2$ satisfy
\begin{align*}
     (U,V_2)   -\!\!\!\!\minuso\!\!\!\!-  (W_1,W_2)  -\!\!\!\!\minuso\!\!\!\!- V_1, \quad  
     (U,W_1,V_1)  -\!\!\!\!\minuso\!\!\!\!-  W_2  -\!\!\!\!\minuso\!\!\!\!- V_2.
 \end{align*}
\end{remark}

Given $\mathcal{Q}_{W_1W_2|U}$, we denote by $\mathcal{Q}_{U|W_1W_2} \in \Delta(\mathcal{U})^{|\mathcal{W}_1\times\mathcal{W}_2|}$ and $\mathcal{Q}_{U|W_2} \in \Delta(\mathcal{U})^{|\mathcal{W}_2|}$ the posterior beliefs of decoders $\mathcal{D}_1$ and $\mathcal{D}_2$. 
Moreover, for $(w_1,w_2) \in \mathcal{W}_1 \times \mathcal{W}_2$, we introduce the notations $\mathcal{Q}^{w_1w_2}_U =\mathcal{Q}_{U|W_1W_2}(.|w_1,w_2) \in \Delta(\mathcal{U})$ and $\mathcal{Q}^{w_2}_U =\mathcal{Q}_{U|W_2}(.|w_2)\in \Delta(\mathcal{U})$.

\begin{theorem}\label{main result}
Let $(R_1,R_2) \in \mathbb{R}_{+}^2$, we have
\begin{align*} 
\forall \varepsilon>0,  \exists \hat{n} \in \mathbb{N},   \forall n \geq \hat{n},  D^n_e(R_1,R_2) &\leq D_e^{\star}(R_1,R_2) + \varepsilon,\\
\forall n \in \mathbb{N},  D_e^n(R_1,R_2) &\geq D_e^{\star}(R_1,R_2).  
\end{align*}
\end{theorem}
The proof of Theorem \ref{main result} is stated in Sec. \ref{sec:converse} and \ref{sec:achievability}. Together with Fekete's Lemma for the sub-additive sequence $\big(n D_e^n(R_1,R_2)\big)_{n\in \N^{\star}}$ (see Lemma \ref{lemma:subadditive}),
we obtain 
\begin{align}
\lim_{n\rightarrow \infty} D^n_e(R_1,R_2)= \underset{n\in\mathbb{N}^{\star}}{\inf} D^n_e(R_1,R_2) = D^{\star}_e(R_1,R_2).
\end{align}

\section{Converse Proof of Theorem \ref{main result}}\label{sec:converse} 

Let $(R_1,R_2) \in \mathbb{R}^2_{+}$ and $n\in\N^{\star}$. We consider $(\sigma,\tau_1,\tau_2)\in\mathcal{S}(n,R_1,R_2)$ and a random variable $T$ uniformly distributed over $\{1,2,...,n\}$ and independent of $(U^n,M_1,M_2,V_1^n,V_2^n)$. We introduce the auxiliary random variables $W_1 =(M_1,T)$, $W_2 =(M_2,T)$, $(U,V_1,V_2)=(U_T,V_{1,T},V_{2,T})$, distributed according to $\mathcal{P}_{UW_1W_2V_1V_2}^{\sigma\tau_1\tau_2}$ defined for all $(u,w_1,w_2,v_1,v_2) = (u_t,m_1,m_2,t,v_{1,t},v_{2,t})$ by
\begin{align}
  &\mathcal{P}_{UW_1W_2V_1V_2}^{\sigma\tau_1\tau_2} (u,w_1,w_2,v_1,v_2)\nonumber\\
  =& \mathcal{P}_{U_TM_1M_2TV_{1T}V_{2T}}^{\sigma\tau_1\tau_2}  (u_t,m_1,m_2,t,v_{1,t},v_{2,t})\nonumber\\
 =&\frac1n\sum_{u^{t-1}\atop u_{t+1}^n}\sum_{v_1^{t-1},v_{1,t+1}^n\atop v_2^{t-1},v_{2,t+1}^n}\!\!\!\!\!\!
\bigg(\prod_{t=1}^n\mathcal{P}_{U}(u_t)\bigg)\sigma(m_1,m_2|u^n)
 \times\tau_1(v_1^n|m_1,m_2)\tau_2(v_2^n|m_2). \label{eq:auxRV_converse}
\end{align}
\begin{lemma}\label{lemma:decomposition}
The distribution $\mathcal{P}_{UW_1W_2V_1V_2}^{\sigma\tau_1\tau_2}$ has marginal on $\Delta(\mathcal{U})$ given by $\mathcal{P}_U$ and satisfies the Markov chain properties
\begin{align*}
     (U,V_2)   -\!\!\!\!\minuso\!\!\!\!-  (W_1,W_2)  -\!\!\!\!\minuso\!\!\!\!- V_1, \quad  
     (U,W_1,V_1)  -\!\!\!\!\minuso\!\!\!\!-  W_2  -\!\!\!\!\minuso\!\!\!\!- V_2.
 \end{align*}
\end{lemma}
\begin{proof}[Lemma \ref{lemma:decomposition}]
The i.i.d. property of the source ensures that the marginal distribution is $\mathcal{P}_U$. By the definition of the decoding functions $\tau_1$ and $\tau_2$ we have 
\begin{align*}
     &(U_T,V_{2,T})   -\!\!\!\!\minuso\!\!\!\!-  (M_1,M_2,T)  -\!\!\!\!\minuso\!\!\!\!- V_{1,T},\\
     &(U_T,M_1,V_{1,T})  -\!\!\!\!\minuso\!\!\!\!-  (M_2,T)  -\!\!\!\!\minuso\!\!\!\!- V_{2,T}.
 \end{align*}
Therefore $\mathcal{P}_{UW_1W_2V_1V_2}^{\sigma\tau_1\tau_2} = \mathcal{P}_U\mathcal{P}_{W_1W_2|U}^{\sigma}\mathcal{P}_{V_1|W_1W_2}^{\tau_1}\mathcal{P}_{V_2|W_2}^{\tau_2}$.
\end{proof}

\begin{lemma}\label{lemma:belongtoQ0}
For all $\sigma$, the distribution $\mathcal{P}_{W_1W_2|U}^{\sigma}\in \mathbb{Q}_0$.
\end{lemma}
\begin{proof}[Lemma \ref{lemma:belongtoQ0}]
We consider an encoding strategy $\sigma$, then 
\begin{align}
&n R_2 \geq H(M_2) \geq I(M_2;U^n)  \label{e:Ide1} \\
&= \sum_{t=1}^n I(U_t;M_2|U^{t-1}) \label{memorylesssourceyyy} 
\\
&= n I(U_T;M_2|U^{T-1},T)\label{e:Ide2} \\
&= n I(U_T;M_2,U^{T-1},T) \label{e:Ide3} 
\\
&\geq n I(U_T;M_2,T)\label{e:Ide4}\\
&= n I(U;W_2).\label{e:Ide5}
\end{align}
In fact, \eqref{e:Ide2} follows from the introduction of the uniform random variable $T\in\{1,\ldots,n\}$, \eqref{e:Ide3} comes 
 from the i.i.d. property of the source, and \eqref{e:Ide5} follows from the identification of the auxiliary random variables $(U,W_2)$ and the independence between $T$ and $U_T$. Similarly,
\begin{align} 
&n(R_1+R_2) \geq  H(M_1,M_2)  \geq I(U^n;M_1,M_2) \label{1lossy source} \nonumber\\ 
=& \sum_{t=1}^nI(U_t;M_1,M_2|U^{t-1}) \\
=& nI(U_T;M_1,M_2|U^{T-1},T) \nonumber \\ 
\geq& nI(U_T;M_1,M_2,T) \\
=& nI(U;W_1,W_2). 
\end{align} \end{proof}

\begin{lemma}\label{lemma:singleletterdistortion}
For all $(\sigma,\tau_1,\tau_2)$ and $i\in\{1,2\}$, we have
   $d_e^n(\sigma,\tau_1,\tau_2) = \E   \big[d_e(U,V_1,V_2)\big] $  and 
   $d_i^n(\sigma,\tau_i) = \E
   \big[d_i(U,V_i)\big]$ 
evaluated with respect to $\mathcal{P}_U\mathcal{P}_{W_1W_2|U}^{\sigma}\mathcal{P}_{V_1|W_1W_2}^{\tau_1}\mathcal{P}_{V_2|W_2}^{\tau_2}$. Moreover, for all $\sigma$, we have
\begin{align}
\mathbb{Q}_{1}(\mathcal{P}^{\sigma}_{W_1W_2|U}) =& \Big\{ \mathcal{Q}_{V_1|W_1W_2},\;
\qquad \exists \tau_1 \in BR_1(\sigma),\; \mathcal{Q}_{V_1|W_1W_2} = \mathcal{P}^{\tau_1}_{V_1|W_1W_2}\Big\},\label{eq:lemmaBRset}\\
\mathbb{Q}_{2}(\mathcal{P}^{\sigma}_{W_2|U}) =& \Big\{ \mathcal{Q}_{V_2|W_2},\;
\qquad \exists \tau_2 \in BR_2(\sigma),\; \mathcal{Q}_{V_2|W_2} = \mathcal{P}^{\tau_2}_{V_2|W_2}\Big\}.\label{eq:lemmaBRsets}
\end{align}
\end{lemma}

\begin{proof}[Lemma \ref{lemma:singleletterdistortion}] By Definition \ref{singlelongrun} and  \eqref{jointproba}, \eqref{eq:auxRV_converse}, we have
\begin{align}
d_e^n(\sigma,\tau_1,\tau_2) =& \sum_{u^n, m_1,m_2,\atop v_1^n,v_2^n}\bigg(\prod_{t=1}^n\mathcal{P}_{U}(u_t)\bigg)\sigma(m_1,m_2|u^n)
\times \tau_1(v_1^n|m_1,m_2)\tau_2(v_2^n|m_2)\Bigg[\frac{1}{n}\sum_{t=1}^n d_e(u_t,v_{1,t},v_{2,t})\Bigg]\nonumber\\
=& \sum_{t=1}^n \sum_{u_t,m_1,m_2,\atop v_{1,t}, v_{2,t}} \mathcal{P}^{\sigma,\tau_1,\tau_2}(u_t,m_1,m_2,t, v_{1,t}, v_{2,t})
\times d_e(u_t,v_{1,t},v_{2,t}) \nonumber\\ 
=&\E\big[d_e(U_T,V_{1,T},V_{2,T})\big]=\E\big[d_e(U,V_1,V_2)\big] .
\end{align} 
Now we prove the second part of lemma \ref{lemma:singleletterdistortion}. For any $\sigma$ and any $\mathcal{Q}_{V_1| W_1W_2}\in \mathbb{Q}_{1}(\mathcal{P}^{\sigma}_{W_1W_2|U})$, we define $\tilde{\tau}_1$ by
\begin{align}
\tilde{\tau}_1(v_1^n|m_1,m_2) = \prod_{s=1}^n \mathcal{Q}_{V_1| W_1W_2}(v_{1,s}|m_1,m_2,s), \label{eq:T}
\end{align} 
$ \forall (m_1,m_2,v_1^n)$. Then $\forall (w_1,w_2,v_1)=(m_1,m_2,t,v_{1,t})$
\begin{align}
&\mathcal{P}^{\tilde{\tau}_1}_{V_1|W_1W_2}(v_1|w_1,w_2) = \mathcal{P}^{\tilde{\tau}_1}_{V_1|W_1W_2}(v_{1,t}|m_1,m_2,t) \nonumber\\
=& \sum_{v_1^{t-1},v_{1,t+1}^n} \tilde{\tau}_1(v_{1}^n|m_1,m_2)\nonumber \\
=& \sum_{v_1^{t-1},v_{1,t+1}^n} \prod_{s=1}^n \mathcal{Q}_{V_1| W_1W_2}(v_{1,s}|m_1,m_2,s)\nonumber\\
=&  \mathcal{Q}_{V_1| W_1W_2}(v_{1,t}|m_1,m_2,t) 
\times 
\sum_{v_1^{t-1},v_{1,t+1}^n} \prod_{s\neq t} \mathcal{Q}_{V_1| W_1W_2}(v_{1,s}|m_1,m_2,s)\nonumber\\
=&  \mathcal{Q}_{V_1| W_1W_2}(v_{1,t}|m_1,m_2,t)= \mathcal{Q}_{V_1| W_1W_2}(v_{1}|w_1,w_2).
\label{eq:T2}
\end{align} 
Moreover assume that $\tilde{\tau}_1 \notin BR_1(\sigma)$, then there exists $\bar{\tau}_1\neq \tilde{\tau}_1$ such that 
\begin{align}
&\E_{\mathcal{P}^{\sigma}_{W_1W_2|U}\atop \mathcal{P}^{\bar{\tau}_1}_{V_1| W_1W_2}}\big[d_1(U,V_1)\big] = d_1^n(\sigma,\bar{\tau}_1)<d_1^n(\sigma,\tilde{\tau}_1)\nonumber\\
&= \E_{\mathcal{P}^{\sigma}_{W_1W_2|U}\atop \mathcal{P}^{\tilde{\tau}_1}_{V_1| W_1W_2}}\big[d_1(U,V_1)\big]= \E_{\mathcal{P}^{\sigma}_{W_1W_2|U}\atop \mathcal{Q}_{V_1| W_1W_2}}\big[d_1(U,V_1)\big],
\end{align} 
which contradicts  $\mathcal{Q}_{V_1| W_1W_2}\in \mathbb{Q}_{1}(\mathcal{P}^{\sigma}_{W_1W_2|U})$. Therefore, $\tilde{\tau}_1\in BR_1(\sigma)$ and thus $\mathcal{Q}_{V_1| W_1W_2}$ belongs to the right-hand side of \eqref{eq:lemmaBRset}. For the other inclusion, we assume that $\mathcal{Q}_{V_1| W_1W_2}$ belongs to the right-hand side of \eqref{eq:lemmaBRset} and does not belong to $\mathbb{Q}_{1}(\mathcal{P}^{\sigma}_{W_1W_2|U})$, then we show that it leads to a contradiction. Similar arguments imply \eqref{eq:lemmaBRsets}. 
\end{proof}

For any strategy $\sigma$, we have
\begin{align}
&\underset{\tau_1 \in BR_1(\sigma),\atop\tau_2 \in BR_2(\sigma)}{\max} \;  d_e^n(\sigma,\tau_1,\tau_2) \nonumber \\ 
=&\underset{\tau_1\in BR_1(\sigma),\atop \tau_2 \in BR_2(\sigma)}{\max} \; \mathbb{E}_{\mathcal{P}^{\sigma}_{W_1W_2|U}\atop\mathcal{P}^{\tau_1}_{V_{1}|W_1W_2}  \mathcal{P}^{\tau_2}_{V_{2}|W_2}}\Big[d_e(U,V_{1},V_{2})\Big] \label{zachi1}\\ 
=&\underset{\mathcal{Q}_{V_{1}|W_1W_2}\in \mathbb{Q}_{1}(\mathcal{P}^{\sigma}_{W_1W_2|U}) \atop \mathcal{Q}_{V_{2}|W_2}\in \mathbb{Q}_{2}(\mathcal{P}^{\sigma}_{W_2|U})}{\max}\mathbb{E}_{\mathcal{P}^{\sigma}_{W_1W_2|U} \atop \mathcal{Q}_{V_{1}|W_1W_2} \ \mathcal{Q}_{V_{2}|W_2}}\Big[d_e(U,V_{1},V_{2})\Big] \label{zachi2}\\ 
\geq&\underset{\mathcal{Q}_{W_1W_2|U}\atop\in\mathbb{Q}_0(R_1,R_2)}{\inf}\underset{\mathcal{Q}_{V_{1}|W_1W_2}\in \mathbb{Q}_{1}(\mathcal{Q}_{W_1W_2|U}) \atop \mathcal{Q}_{V_{2}|W_2}\in \mathbb{Q}_{2}(\mathcal{Q}_{W_2|U})}{\max}\mathbb{E}
\Big[d_e(U,V_{1},V_{2})\Big] \label{zachi3}\\ 
=&D^{\star}_e(R_1,R_2). \label{optdistoooo}
\end{align} 
Equations \eqref{zachi1} and \eqref{zachi2} come from Lemma \ref{lemma:singleletterdistortion}, whereas \eqref{zachi3} comes from Lemma \ref{lemma:belongtoQ0} and \eqref{optdistoooo} follows from \eqref{optdistoooo1}. Therefore, $D^{\star}_e(R_1,R_2) \leq \underset{\sigma}{\inf}\underset{\tau_1 \in BR_{1}(\sigma), \atop \tau_2 \in BR_{2}(\sigma) }{\max} d_e^n(\sigma,\tau_1,\tau_2) = D_e^n(R_1,R_2)$, $\forall n$. 

\section{Sketch of Achievability Proof of Theorem \ref{main result}}\label{sec:achievability} 

\subsection{Alternative Formulation}

\begin{definition}  For any distributions $q_1 \in \Delta(\mathcal{U})$ and $q_2 \in \Delta(\mathcal{U})$, we denote by $V_1^{\star}(q_1)$ and $V_2^{\star}(q_2)$, the sets of optimal actions of decoders $\mathcal{D}_1$ and $\mathcal{D}_2$. 
\begin{align}
    V_1^{\star}(q_1)=\argmin_{v_1 \in V_1}\sum_{u}q_1(u)d_1(u,v_1),\\
    V_2^{\star}(q_2)=\argmin_{v_2 \in V_2}\sum_{u}q_2(u)d_2(u,v_2).
\end{align}
\end{definition}

\begin{definition} 
Given a strategy $\mathcal{Q}_{W_1W_2|U}$ and symbols $(w_{1},w_{2})$, we denote by $\mathcal{Q}_{U}^{w_{1}w_{2}}\in \Delta(\mathcal{U})$ and $\mathcal{Q}_{U}^{w_{2}}\in \Delta(\mathcal{U})$ the Bayesian posterior beliefs
defined by
\begin{align}
\mathcal{Q}_{U}^{w_{1}w_{2}}(u) =& \frac{\mathcal{P}_U(u)\mathcal{Q}_{W_1W_2|U}(w_1,w_2|u)}{\sum_{u'}\mathcal{P}_U(u')\mathcal{Q}_{W_1W_2|U}(w_1,w_2|u')}.
\end{align}
Among the set of optimal actions of $\mathcal{D}_1$ and $\mathcal{D}_2$, we denote the worst pairs for the encoder distortion by
\begin{align}
    \tilde{A}(\mathcal{Q}_{W_1W_2|U},w_{1},w_{2}) 
    = 
    \underset{(v_1,v_2) \in V_1^{\star}(\mathcal{Q}_{U}^{w_{1}w_{2}}) \times \atop V_2^{\star}(\mathcal{Q}_{U}^{w_{2}})}{\argmax}\Big\{\sum_{u}\mathcal{Q}_U^{w_{1},w_{2}}(u)d_e(u,v_1,v_2)\Big\}.
\end{align}
\end{definition}
\begin{definition}
Given $(R_1,R_2) \in \mathbb{R}^{2}_{+}$, we define 
\begin{align}
&\tilde{\mathbb{Q}}_0(R_1,R_2) = \Big\{\mathcal{Q}_{W_1W_2|U} \ s.t. \  R_2 > I(U;W_2) \ , \nonumber \\ 
&R_1+R_2 > I(U;W_1,W_2), 
\max_{w_1,w_2}| \tilde{A}(\mathcal{Q}_{W_1W_2|U},w_{1},w_{2})|=1 \Big\}. \nonumber
\end{align}
\end{definition}
\begin{definition} Consider the following problem
 \begin{align}
 \tilde{D}_e(R_1,R_2)= 
 \underset{\mathcal{Q}_{W_1W_2|U}\atop \in\tilde{\mathbb{Q}}_0(R_1,R_2)}{\inf}\underset{\mathcal{Q}_{V_1|W_1W_2}\in \mathbb{Q}_{1}(\mathcal{Q}_{W_1W_2|U}) \atop \mathcal{Q}_{V_2|W_2}\in \mathbb{Q}_{2}(\mathcal{Q}_{W_2|U})}{\max}\mathbb{E}\Big[
    d_e(U,V_1,V_2)\Big], 
    \nonumber
 \end{align} 
 where the expectation 
 is evaluated with respect to $\mathcal{P}_{U} \mathcal{Q}_{W_1W_2|U}\mathcal{Q}_{V_1|W_1W_2}\mathcal{Q}_{V_2|W_2}$. \end{definition}

\begin{lemma}
\label{lemmmaa5} For $(R_1,R_2) \in \mathbb{R}_{+}^2$, 
  $ D_e^{\star}(R_1,R_2) = \tilde{D}_e(R_1,R_2).$
\end{lemma}
Similarly to the proof of \cite[Lemma A.5]{jet}, this proof relies on showing that $\tilde{\mathbb{Q}}_0(R_1,R_2)$ is dense in $\mathbb{Q}_0(R_1,R_2)$. It is provided in the full version of the paper \cite[Lemma 6]{rouphael2021strategic}.

\subsection{Achievability Scheme}

\vspace{-0.5cm}

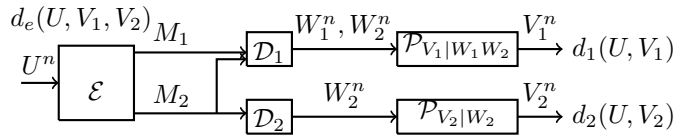
\begin{figure}[h!]
    \centering
    \begin{tikzpicture}[xscale=2,yscale=0.9]
\draw[thick,->](0.5,0.25)--(0.75,0.25);
\draw[thick,-](0.75,-0.3)--(1.25,-0.3)--(1.25,0.8)--(0.75,0.8)--(0.75,-0.3);
\node[below, black] at (1,0.4) {$\mathcal{E}$};
\draw[thick, ->](1.25,0.7)--(2,0.7);
\draw[thick, ->](1.8,-0.2)--(1.8,0.6)--(2,0.6);
\draw[thick, ->](1.25,-0.2)--(2,-0.2);
\draw[thick,-](2,0.5)--(2.30,0.5)--(2.30,1)--(2,1)--(2,0.5);
\node[below, black] at (2.15,1) { $\mathcal{D}_{1}$};
\draw[thick,-](2,-0.5)--(2.30,-0.5)--(2.30,0)--(2,0)--(2,-0.5);
\node[below, black] at (2.15,0) { $\mathcal{D}_{2}$};
\node[above, black] at (0.62,0.25) {$U^n$};
\node[above, black] at (1.5,0.67){$M_1$};
\node[above, black] at (1.5,-0.25){$M_2$};
  \draw[thick,-](3,-0.5)--(3.8,-0.5)--(3.8,0)--(3,0)--(3,-0.5);
\draw[thick,-](3,0.5)--(3.8,0.5)--(3.8,1)--(3,1)--(3,0.5);
\draw[thick,->](2.30,0.75)--(3,0.75);
\node[above, black] at (3.95,0.75) {$V_1^n$};
\node[above, black] at (3.95,-0.25) {$V_2^n$};
\draw[thick,->](3.8,0.75)--(4.1,0.75);
\draw[thick,->](3.8,-0.25)--(4.1,-0.25);
\node[above, black] at (3.4,0.45) {$\mathcal{P}_{V_1|W_1W_2}$};
\node[above, black] at (3.4,-0.55) {$\mathcal{P}_{V_2|W_2}$};
\draw[thick,->](2.30,-0.25)--(3,-0.25);
\node[above, black] at (0.9,0.85) {$d_e(U,V_1,V_2)$};
\node[above, black] at (2.64,0.75) {$W_1^n,W_2^n$};
\node[above, black] at (2.65,-0.25) {$W_2^n$};
\node[right, black] at (4.1,-0.25) {$d_2(U,V_2)$};
\node[right, black] at (4.1,0.75) {$d_1(U,V_1)$};
    \end{tikzpicture}
    \caption{Achievability of successive refinement source coding.}
    \label{fig:successive}
\end{figure}
\subsubsection{$R_1>0, R_2>0$ 
}
Fix a conditional probability distribution $\mathcal{Q}_{W_1,W_2|U}$. There exists $\eta>0$ such that \begin{align}
    R_2 = I(U;W_2) + \eta, \quad \label{eqwwe}
    R_1 = I(U;W_1|W_2) + \eta. 
\end{align} 
Codebook generation: Randomly and independently generate $2^{\lfloor nR_{2}\rfloor}$ sequences $w_2^n(m_{2})$ for $m_{2} \in [1:2^{\lfloor{nR_{2}}\rfloor}]$, according to the i.i.d distribution $\mathcal{P}_{W_2^n}=\Pi_{t=1}^n\mathcal{P}_{W_2}(w_{2t})$. For each $(m_{1},m_{2}) \in [1:2^{\lfloor nR_{1} \rfloor}] \times [1:2^{\lfloor nR_{2} \rfloor}]$ generate a sequence 
$w_1^n(m_{1},m_{2})$ randomly and conditionally independently according to the i.i.d conditional distribution  $\mathcal{P}_{W_1^n|M_1W_2^n}=\Pi_{t=1}^n\mathcal{P}_{W_1|M_1W_2}(w_{1t}|m_1,w_{2t}(m_{2}))$. 

Encoding strategy $\sigma$: Encoder $\mathcal{E}$ observes $u^n$ and looks in the codebook for a pair $(m_{1},m_{2})$ such that \\ $(u^n, w_1^n(m_{1},m_{2}),w^n_2(m_{2})) \in \mathcal{T}_{\delta}^n(\mathcal{P}_U\mathcal{P}_{W_1W_2|U})$, i.e. the sequences are jointly typical with tolerance parameter $\delta>0$. If such a jointly typical tuple doesn't exist, the source encoder sets $(m_{1},m_{2})$ to $(1,1)$.
Then, it sends  $m_{2}$ to decoder $\mathcal{D}_{2}$, and $(m_{1},m_{2})$ to decoder $\mathcal{D}_{1}$. 

Here comes the main difference with the successive refinement coding, which is due to the strategic nature of the problem. Instead of declaring $w_1^n(m_{1},m_{2})$ and $w_2^n(m_{2})$ and selecting $V_1^n$ and $V_2^n$ i.i.d. with respect to $\mathcal{Q}_{V_1|W_1W_2}\in \mathbb{Q}_{1}(\mathcal{Q}_{W_1W_2|U})$ and $ \mathcal{Q}_{V_2|W_2}\in \mathbb{Q}_{2}(\mathcal{Q}_{W_2|U})$, at each stage $t\in\{1,\ldots,n\}$ the decoders $\mathcal{D}_1$ and $\mathcal{D}_2$ compute their Bayesian posterior beliefs $\mathcal{P}^{\sigma}_{U_t|M_1M_2}(\cdot|m_1,m_2)$ and $\mathcal{P}^{\sigma}_{U_t|M_2}(\cdot|m_2)$ and select the actions  $v_{1,t}\in V_1^{\star}(\mathcal{P}^{\sigma}_{U_t|M_1M_2})$ and $v_{2,t}\in V_2^{\star}(\mathcal{P}^{\sigma}_{U_t|M_2})$ that minimize their own distortion function. If several pairs are available, they select the worst one for the encoder distortion.

Error Event: Given a tolerance $\delta>0$, the error event is given by $\mathcal{F}=\{(U^n,W^n_2(m_{2}),W_1^n(m_{2},m_1) \notin \mathcal{T}_{\delta}^n \}$. We have by the union of events bound $\mathcal{P}(\mathcal{F})\leq \mathcal{P}(\mathcal{F}_{1}) + \mathcal{P}(\mathcal{F}_{2}(M_2) \cap \mathcal{F}_{1}^c)$, where 
$\mathcal{F}_{1} = \{(U^n,W_2^n(m_{2})) \notin \mathcal{T}_{\delta}^n \  \forall m_{2}\}$, 
$\mathcal{F}_{2}(m_2) = \{(U^n,W^n_2(m_{2}),W_1^n(m_{2},m_1)) \notin \mathcal{T}_{\delta}^n \ \forall m_{1}\}$. 
By \cite[Lemma 3.3, pp. 62]{elgamal},  $\mathcal{P}(\mathcal{F}_{1})$ tends to zero as $n \to \infty$ if 
$R_{2} > \  I(U;W_2) + \eta.$ 
By \cite[Lemma 3.3, pp. 62]{elgamal}, $\mathcal{P}(\mathcal{F}_{1}^c \cap \mathcal{F}_{2}(M_2))$ goes to zero if 
$R_1 + R_{2} > \ I(U;W_1,W_2)  +\eta.$ 

Since the expected error probability evaluated with respect to the random codebook is small, we have that for all $\varepsilon_2 >0$, for all $\eta>0$, there exists $\Bar{\delta}>0$, for all $\delta\leq\Bar{\delta}$, there exists $\Bar{n} \in \mathbb{N}$ such that for all $n \geq \Bar{n}$, we have
\begin{align}
&\mathbb{E}\big[\mathcal{P}(\mathcal{F}_{1})\big] \leq \varepsilon_2, \label{eqarr} \quad
&\mathbb{E}\big[\mathcal{P}(\mathcal{F}_{2}(m_2))\big] \leq \varepsilon_2.  
\end{align} 

\subsubsection{Control of beliefs} 
We introduce the indicator of error events $E^1_{\delta} \in \{0,1\}$ for decoder $\mathcal{D}_1$
defined as follows
\begin{align}
  E^1_{\delta}=&\begin{cases}
    1, & \text{if $(u^n,w_1^n,w_2^n) \notin \mathcal{T}_{\delta}^n(\mathcal{P}_{U}\mathcal{Q}_{W_1W_2|U}) $}.\\
    0, & \text{otherwise}.
  \end{cases} 
\end{align}

We denote the Bayesian posterior beliefs $\mathcal{P}^{\sigma}_{U_t|M_1M_2}(\cdot|m_1,m_2)\in\Delta(\mathcal{U})$ and $\mathcal{P}^{\sigma}_{U_t|M_2}(\cdot|m_2)\in\Delta(\mathcal{U})$ by $\mathcal{P}^{m_1,m_2}_{t}$ and $\mathcal{P}^{m_2}_{t}$. 
We show that on average, the Bayesian beliefs are close in KL distance to the target beliefs $\mathcal{Q}_{U|W_1W_2}$ and $\mathcal{Q}_{U|W_2}$ induced by the single-letter distribution $\mathcal{Q}_{W_1W_2|U}$. Assuming the distribution $\mathcal{Q}_{U|W_1W_2}$ is fully supported, the beliefs of decoder $\mathcal{D}_{1}$ are controlled as follows

\begin{align}
       &\mathbb{E} \Big[\frac{1}{n}\sum_{t=1}^n D(\mathcal{P}^{m_1,m_2}_{t} ||\mathcal{Q}_{U|W_{1}W_{2}}(\cdot | W_{1t},W_{2t})) \Big|E^1_{\delta}=0\Big] \nonumber \\
     =& \sum_{m_1,m_2,\atop w_1^n,w_2^n}\mathcal{P}^{\sigma,\tau_1,\tau_2}(m_1,m_2,w_1^n,w_2^n\Big| E^1_{\delta}=0) \nonumber \\ &\times \frac{1}{n}\sum_{t=1}^n\sum_{u}\mathcal{P}_{t}^{m_1m_2}(u)\log_2\frac{\mathcal{P}_{t}^{m_1m_2}(u)}{\mathcal{Q}_{U|W_{1}W_{2}}(u| w_{1t},w_{2t})} 
     \nonumber 
     \\
     =& \sum_{m_1,m_2,\atop w_1^n,w_2^n}\mathcal{P}^{\sigma,\tau_1,\tau_2}(m_1,m_2,w_1^n,w_2^n\Big| E^1_{\delta}=0) \nonumber \\ &\times \frac{1}{n}\sum_{t=1}^n\sum_{u}\mathcal{P}_{t}^{m_1m_2}(u)\log_2 \frac{1}{\mathcal{Q}_{U|W_{1}W_{2}}(u| w_{1t},w_{2t})} \nonumber \\
     &- \sum_{m_1,m_2,\atop w_1^n,w_2^n}\mathcal{P}^{\sigma,\tau_1,\tau_2}(m_1,m_2,w_1^n,w_2^n\Big| E^1_{\delta}=0)\times \nonumber \\ &\frac{1}{n}\sum_{t=1}^n\sum_{u}\mathcal{P}_{t}^{m_1m_2}(u)\log_2\frac{1}{\mathcal{P}_{t}^{m_1m_2}(u)} 
     \nonumber \\
    \leq& \frac{1}{n}I(U^n;M_1,M_2\Big| E^1_{\delta}=0)-I(U;W_1,W_2)+\delta \nonumber \\ &+\frac{1}{n}+\log_2|\mathcal{U}|\cdot\mathcal{P}^{\sigma,\tau_1,\tau_2}(E^1_{\delta}=1) 
    \nonumber \\
   \leq& \eta + \delta +\frac{1}{n}+\log_2|\mathcal{U}|\cdot\mathcal{P}^{\sigma,\tau_1,\tau_2}(E^1_{\delta}=1)\label{b7}. 
\end{align}

\subsubsection{Conclusion} By combining the equations \eqref{eqarr}, 
\eqref{b7} with \cite[Lemma A.21, equations (40)-(46), Lemma A.8 ]{jet}, we obtain $\forall \varepsilon>0$, $\exists \hat{n}$, $\forall n \geq\hat{n}$,  $D_e^n(R_1,R_2)\leq D^{\star}_e(R_1,R_2)+\varepsilon$. More details are provided in Appendix \ref{Appendix B}.

\subsection{Special Cases}




\subsubsection{$R_1=R_2=0$} 
The auxiliary random variables $(W_1,W_2)$ are independent of $U$. The message sets are singletons, and the only possible encoding strategy $\sigma_0$ is given by $\sigma_0:\mathcal{U}^n\longrightarrow\{1\}\times\{1\}$. The codebook consists of two sequences $W_2^n(1)$ and $W_1^n(1,1)$ only.
Therefore, $\forall n \in \mathbb{N}^{\star}$, 
 $D_e^{\star}(0,0)=D_e^n(0,0).$ 
 \subsubsection{$R_1>0 \And R_2=0$}

 Random variables $W_2$ and $U$ are independent for $R_1>0$ and $R_2=0$, i.e. $\mathcal{Q}_{W_1W_2|U}=\mathcal{Q}_{W_2}\mathcal{Q}_{W_1|W_2U}$. This means that decoder $\mathcal{D}_2$ will repeatedly chose the action $v_{2,0} \in V^{\star}(\mathcal{P}_U)$ that corresponds to its prior belief $\mathcal{P}_U$ and maximizes the encoder's distortion. The persuasion game is thus reduced to the point-to-point problem with one decoder $\mathcal{D}_1$, as in \cite{jet}. 

\subsubsection{$R_1=0 \And R_2>0$ } 

The auxiliary random variable $W_1$ is independent of $U$. Hence, the encoder transmits the same index to both decoders. Therefore, both decoders will have the same posterior belief $\mathcal{Q}^{w_2}_U \in \Delta(\mathcal{U})$, $\forall w_2 \in \mathcal{W}_2$. 


In that case, the optimal distortion can be reformulated in terms of a convexification of its expected distortion as in \cite{jet}, 
$ D^{\star}_e(0,R_2) = \underset{(\lambda_{w_2},\mathcal{Q}_U^{w_2})_{w_2\in\mathcal{W}_2}}{\inf}\sum_{w_2\in\mathcal{W}_2}\lambda_{w_2}\Psi_e(\mathcal{Q}_U^{w_2})$ 
where $\Psi_e(q)=\underset{(v_1,v_2) \in \atop V_1^{\star}(q) \times  V_2^{\star}(q)}{\max}\mathbb{E}_q\big[d_e(U,v_1,v_2)\big].$

\appendices{}


\section{Proof of Lemma \ref{lemma:subadditive}}
\label{Appendix A} 
\begin{proof}[Lemma \ref{lemma:subadditive}]
Let $n,m \in \mathbb{Z}$. We denote by ${\sigma}_c^{n+m}$, the concatenation of the strategies ${\sigma}^n$, $\sigma^m$ where ${\sigma}^n$ is implemented over the first $n$ stages and ${\sigma}^m$ is implemented over the last $m$ stages. For decoder $i \in \{1,2\}$, consider the best responses ${\tau_i}^n \in BR_{i}({\sigma}^n)$ and ${\tau_i}^m \in BR_{i}({\sigma}^m)$. Then, the concatenation $\tau_{i,c}^{n+m}$ of ${\tau_i}^n$ and ${\tau_i}^m$ is also a best response $\tau_{i,c}^{n+m}\in BR_i({\sigma}_c^{n+m})$. Therefore, we have the inequality
\begin{align}
&n D_e^n(R_1,R_2) + m D_e^m(R_1,R_2)\nonumber\\
=&\underset{\sigma^{n}}{\inf}\underset{{\tau_1}^{n} \in BR_{1}(\sigma^{n}), \atop {\tau_2}^{n} \in BR_{2}(\sigma^{n}) }{\max} \mathbb{E}
  \Big[\sum_{t=1}^{n} d_e(U_t,V_{1,t},V_{2,t})\Big] \nonumber \\ 
  &+ \underset{\sigma^{m}}{\inf}\underset{{\tau_1}^{m} \in BR_{1}(\sigma^{m}), \atop {\tau_2}_{m} \in BR_{2}(\sigma^{m}) }{\max} \mathbb{E}
  \Big[\sum_{t=1}^{m} d_e(U_t,V_{1,t},V_{2,t})\Big] \\
=&\underset{\sigma^{n+m}_c}{\inf}\underset{\tau_{1}^{n+m} \in BR_{1}(\sigma_c^{n+m}), \atop \tau_{2}^{n+m} \in BR_{2}(\sigma_c^{n+m}) }{\max} \mathbb{E}
  \Big[\sum_{t=1}^{n+m} d_e(U_t,V_{1,t},V_{2,t})\Big] \\
\geq&\underset{\sigma^{n+m}}{\inf}\underset{{\tau_1}^{n+m} \in BR_{1}(\sigma^{n+m}), \atop {\tau_2}^{n+m} \in BR_{2}(\sigma^{n+m}) }{\max} \mathbb{E}
  \Big[\sum_{t=1}^{n+m} d_e(U_t,V_{1,t},V_{2,t})\Big] \\
=&(n+m) D_e^{n+m}(R_1,R_2),
\end{align}
where the notation $\sigma^{n+m}_c$
stands for the encoding strategies obtained by concatenation.\end{proof}

\ifCLASSINFOpdf
\else
\fi

\definition [KL Divergence] The Kullback-Leiber (KL) Divergence for distributions $P$ and $Q$ on $\Delta(U)$ with respective supports $\mathrm{supp} P$ and $\mathrm{supp} Q$ is given by
\begin{align}
    D(P || Q) = \begin{cases} \sum_{u \in \mathrm{supp} P}P(u)\log_2\frac{P(u)}{Q(u)}&, \text{ if $\mathrm{supp} Q \ \subset \  \mathrm{supp} P$}. \\
   + \infty &, \text{ otherwise.} 
      \end{cases}
\end{align} 
 \definition [Typical Sequences]
Let $\mathcal{X}$ be a finite alphabet and $x^n$ a sequence in $\mathcal{X}^n,$ and let $\pi_{x^n}$ the empirical probability mass function over $\mathcal{X}$ corresponding to the relative frequency of symbols in $x^n$, i.e. $\pi_{x^n}(x)=\frac{|t:x_t=x|}{n}$ for $x \in \mathcal{X}.$  \\
The sequence $x^n$ is said to be $\delta-$typical with respect to a probability distribution $P_X$ on $\mathcal{X}$ if \begin{align}
    \sum_{x \in X}|\pi_{x^n}(x)-\mathcal{P}_X(x)|\leq\delta.
\end{align} 
We denote by $\mathcal{T}_{\delta}^n(P_X)$ the set of all $\delta-$typical sequences corresponding to $P_X.$ This definition can be extended to $K-$tuples of sequences $(x_1^n,x_2^n,...x_k^n) \in \mathcal{X}_1^n \times \mathcal{X}_2^n \times ... \times \mathcal{X}_k^n,$ that are jointly $\delta-$typical with respect to the joint probability $P_{X_1...X_k}.$ The set of all such $k-$tuples is denoted by $\mathcal{T}_{\delta}^n(\mathcal{P}_{X_1...X_k}).$
\section{Proof of Achievability of Theorem \ref{main result} }
\label{Appendix B}
\subsubsection{Alternative Formulation}
\definition  We denote by $V_1^{\star}(q_1)$ and $V_2^{\star}(q_2)$, the respective action sets of decoders $\mathcal{D}_1$ and $\mathcal{D}_2$ for belief parameters $q_1 \in \Delta(\mathcal{U})$ and $q_2 \in \Delta(\mathcal{U})$. 
\begin{align}
    V_1^{\star}(q_1)=\argmin_{v_1 \in V_1}\sum_{u}q_1(u)d_1(u,v_1),\\
    V_2^{\star}(q_2)=\argmin_{v_2 \in V_2}\sum_{u}q_2(u)d_2(u,v_2).
\end{align}

\definition Fix a strategy $\mathcal{Q}_{W_1W_2|U}$. Let $\tilde{A}(\mathcal{Q}_{W_1W_2|U},w_{1},w_{2})$ 
denote the set of action pairs $(v_1,v_2)$ that are optimal for the decoders and worst for the encoder. This set is given by:
\begin{align}
    \tilde{A}(\mathcal{Q}_{W_1W_2|U},w_{1},w_{2})=\underset{(v_1,v_2) \in V_1^{\star}(\mathcal{Q}_{U}^{w_{1}w_{2}}) \times \atop V_2^{\star}(\mathcal{Q}_{U}^{w_{2}})}{\argmax}\Big\{\sum_{u}\mathcal{Q}^{w_{1},w_{2}}(u)d_e(u,v_1,v_2)\Big\} \subset \mathcal{V}_1\times\mathcal{V}_2.
\end{align}


The set $\tilde{\mathbb{Q}}_0(R_1,R_2)$ of target probability distributions for 
$(R_1,R_2) \in \mathbb{R}^{2}_{+}$ is given by:
    \begin{equation}
\tilde{\mathbb{Q}}_0(R_1,R_2) = \Big\{
\mathcal{Q}_{W_1W_2|U} \ s.t. \  R_2 > I(U;W_2) \ , \ 
    R_1+R_2 > I(U;W_1,W_2), \ \max_{w_1,w_2}| \tilde{A}(\mathcal{Q}_{W_1W_2|U},w_{1},w_{2})
    |=1 \Big\}.
   \label{zzsrqq0}
    \end{equation}

\definition Consider the following program:
 \begin{align}
 {
 {\tilde{D}_e(R_1,R_2)=\underset{\mathcal{Q}_{W_1W_2|U}\in\tilde{\mathbb{Q}}_0(R_1,R_2)}{\inf}\underset{\mathcal{Q}_{V_1|W_1W_2}\in \mathbb{Q}_{1}(\mathcal{Q}_{W_1W_2|U}) \atop \mathcal{Q}_{V_2|W_2}\in \mathbb{Q}_{2}(\mathcal{Q}_{W_2|U})}{\max}\mathbb{E}_{\mathcal{P}_{U} \mathcal{Q}_{W_1W_2|U}\atop \mathcal{Q}_{V_1|W_1W_2}\mathcal{Q}_{V_2|W_2}}\Bigg[
    d_e(U,V_1,V_2)\Bigg].}} \label{optdisto}
\end{align} 

\lemma \label{lemmmaa5} For $(R_1,R_2) \in \mathbb{R}_{+}^2$, we have
\begin{align}
   D_e^{\star}(R_1,R_2) = \tilde{D}_e(R_1,R_2)
\end{align}
\proof{of lemma \ref{lemmmaa5}}
Consider the following sets:
\begin{align}
    \mathbb{Q}_0(R_1,R_2)=&\{
\mathcal{Q}_{W_1W_2|U} \ s.t. \  R_2 \geq I(U;W_2) \ , \ 
    R_1+R_2 \geq I(U;W_1,W_2) \}, \\
    \mathbb{Q}_{01}=&\{
\mathcal{Q}_{W_1W_2|U} \ s.t. \ \max_{w_1,w_2} \  | \tilde{A}(\mathcal{Q}_{W_1W_2|U},w_{1},w_{2})
|=1 \}, \\
    \mathbb{Q}_{02}(R_1,R_2)=&\{
\mathcal{Q}_{W_1W_2|U} \ s.t. \  R_2 > I(U;W_2) \ , \ 
    R_1+R_2 > I(U;W_1,W_2) \}.
    \end{align}

We will show that $\mathbb{Q}_{01} \cap \mathbb{Q}_{02}(R_1,R_2) = \tilde{\mathbb{Q}}_0(R_1,R_2)$ is dense in $\mathbb{Q}_0(R_1,R_2)$.
We first show that  $\mathbb{Q}_{01} \cap \mathbb{Q}_0(R_1,R_2)$ is open and dense in $\mathbb{Q}_0(R_1,R_2)$. 
\definition [Equivalent actions]
Two action $v_{i}$ and $\tilde{v}_{i}$ for decoder $\mathcal{D}_i, \ i \in \{1,2\}$ are said to be equivalent if:
$d_i(u,v_{i})=d_i(u,\tilde{v}_{i})$ for all $u\in\mathcal{U},\ \ i \in \{1,2\}$. 
We denote this equivalence relation by $\sim_{i}$. We use $\nsim_{i}$ for non equivalent actions $v_{i}$ and $\tilde{v}_{i}$, i.e. there exists $u\in\mathcal{U}$, such that $d_i(u,v_{i}) \neq d_i(u,\tilde{v}_{i})$ for $i \in \{1,2\}$.\\
Two action pairs $(v_1,v_2)$ and $(\tilde{v}_1,\tilde{v}_2)$ are equivalent for the encoder $\mathcal{E}$ if : $d_e(u,v_1,v_2)=d_e(u,\tilde{v}_1,\tilde{v}_2)$ for all $u\in\mathcal{U}$. We denote this equivalence relation by $\sim_{e}$. We use $\nsim_{e}$ for non equivalent action pairs $(v_1,v_2)$ and $(\tilde{v}_1,\tilde{v}_2)$ i.e. there exists $u\in\mathcal{U}$, such that $d_e(u,v_1,v_2) \neq d_e(u,\tilde{v}_1,\tilde{v}_2)$.
We say that two pairs of actions $(v_1,v_2)$ and $(\tilde{v}_1,\tilde{v}_2)$ are completely equivalent if: \begin{enumerate}
    \item $(v_1,v_2) \sim_e (\tilde{v}_1,\tilde{v}_2)$,
    \item $v_1 \sim_1 \tilde{v}_1$,
    \item $v_2 \sim_2 \tilde{v}_2$.
\end{enumerate}
Without loss of generality we can assume that no pairs of actions are completely equivalent, otherwise we can merge them into one action and reduce the set of actions.
\definition For a fixed i.i.d distribution $\mathcal{P}_U \in \Delta(\mathcal{U})$, we denote by $\mathbb{Q}_{i}$ for $i \in \{1,2\}$, the set of distributions $\mathcal{Q}_{W_1W_2|U} \in \Delta(\mathcal{W}_1\times\mathcal{W}_2)^{|\mathcal{U}|}$ for which decoder $\mathcal{D}_i$ is indifferent between two actions $v_i$ and $\tilde{v}_i$ that are not equivalent, 
\begin{align}
    \mathbb{Q}_{1}=\Bigg\{\mathcal{Q}_{W_1W_2|U}
    , &\exists v_{1} \nsim_{1} \tilde{v}_{1}, \exists w_1,w_2, \nonumber \\  &\mathbb{E}_{\mathcal{Q}_{U}^{w_1w_2}}[d_1(U,v_{1}) ]=\mathbb{E}_{\mathcal{Q}_{U}^{w_1w_2}}[(d_1(U,\tilde{v}_{1})] \Bigg\}, \\
    \mathbb{Q}_{2}=\Bigg\{\mathcal{Q}_{W_1W_2|U}
    , &\exists v_{2} \nsim_{2} \tilde{v}_{2}, \exists w_2, \nonumber \\  &\mathbb{E}_{\mathcal{Q}_{U}^{w_2}}[d_2(U,v_{2}) ]=\mathbb{E}_{\mathcal{Q}_{U}^{w_2}}[(d_2(U,\tilde{v}_{2})] \Bigg\},
\end{align}
and by $\mathbb{Q}_e$, the set of distributions $\mathcal{Q}_{W_1W_2|U} \in \Delta(\mathcal{W}_1\times\mathcal{W}_2)^{|\mathcal{U}|}$ for which the encoder $\mathcal{E}$ is indifferent between two action pairs $(v_1,v_2)$ and $(\tilde{v}_1,\tilde{v}_2)$ that are not equivalent:
\begin{align}
     \mathbb{Q}_e=\Bigg\{\mathcal{Q}_{W_1W_2|U}
   , &\exists (v_{1},v_{2}) \nsim_e (\tilde{v}_{1},\tilde{v}_{2}), \exists w_1,w_2 \nonumber \\    &\mathbb{E}_{\mathcal{Q}_{U}^{w_1w_2}}[d_e(U,v_{1},v_{2})] =\mathbb{E}_{\mathcal{Q}_{U}^{w_1w_2}}[d_e(U,\tilde{v}_{1},\tilde{v}_{2})] \Bigg\}.
\end{align}

Let $\mathbb{Q}^c=\Delta(\mathcal{W}_1\times\mathcal{W}_2)^{|\mathcal{U}|} \backslash \Big(\mathbb{Q}_e \cup \mathbb{Q}_{1} \cup \mathbb{Q}_{2} \Big)$ the set of distributions $\mathcal{Q}_{W_1W_2|U}$ where for all $w_1,w_2$, at least one of the following statements hold: i) The encoder is not indifferent between any two pairs of actions, ii) At least one of the decoders is not indifferent between any two actions. 
\lemma \label{singletonqc} For each distribution $\mathcal{Q}_{W_1W_2|U}$ in $\mathbb{Q}^c$,  the set $\tilde{A}(\mathcal{Q}_{W_1W_2|U},w_{1,t},w_{2,t})$ is a singleton. \\
\proof{of lemma \ref{singletonqc}} 
We proceed by contradiction. 
Let $\mathcal{Q}_{W_1W_2|U} \in \mathbb{Q}^c$ and suppose there exists $(w_1,w_2) \in \mathcal{W}_1\times\mathcal{W}_2$ such that  $|\tilde{A}(\mathcal{Q}_{W_1W_2|U},w_{1},w_{2})| =2
$. This means there exists two distinct action pairs $(v_{1},v_{2}) \neq (\tilde{v}_{1},\tilde{v}_{2})$ with $v_1, \tilde{v}_1 \in V^{\star}(\mathcal{Q}^{w_1w_2}_U)$ and $v_2, \tilde{v}_2 \in V^{\star}(\mathcal{Q}^{w_2}_U)$ such that: \begin{align}
   \mathbb{E}_{\mathcal{Q}_{U}^{w_1w_2}}[d_e(U,v_{1},v_{2})] =&\mathbb{E}_{\mathcal{Q}_{U}^{w_1w_2}}[d_e(U,\tilde{v}_{1},\tilde{v}_{2})],\\
   \mathbb{E}_{\mathcal{Q}_{U}^{w_1w_2}}[d_1(U,v_{1}) ]=&\mathbb{E}_{\mathcal{Q}_{U}^{w_1w_2}}[(d_1(U,\tilde{v}_{1})], \\
\mathbb{E}_{\mathcal{Q}_{U}^{w_2}}[d_2(U,v_{2}) ]=&\mathbb{E}_{\mathcal{Q}_{U}^{w_2}}[(d_2(U,\tilde{v}_{2})]. 
\end{align} 
By hypothesis, $(v_{1},v_{2})$ and $(\tilde{v}_{1},\tilde{v}_{2})$ are not completely equivalent. Therefore, we must have either $(v_{1},v_{2}) \nsim_e (\tilde{v}_{1},\tilde{v}_{2})$, or $v_1 \nsim_1 \tilde{v}_1$, or $v_2 \nsim_2 \tilde{v}_2$, which imply that  $\mathcal{Q}_{W_1W_2|U} \in \mathbb{Q}_e \cup \mathbb{Q}_{1} \cup \mathbb{Q}_{2}$. This contradicts the hypothesis $\mathcal{Q}_{W_1W_2|U} \in \mathbb{Q}^c$. Thus,   
$\tilde{A}(\mathcal{Q}_{W_1W_2|U},w_{1,t},w_{2,t})
$ is a singleton.

\endproof{}
\lemma \label{lemmmaa}The set $\mathbb{Q}^c$ is open and dense in $\Delta(\mathcal{U})$. 

\proof{of lemma \ref{lemmmaa}} 
For each  $v_{i} \nsim_{i} \tilde{v}_{i}, \ i \in \{1,2\} $, and pairs $(v_{1},v_{2}) \nsim_{e} (\tilde{v}_{1},\tilde{v}_{2})$ each  set 
\begin{align}
    \mathbb{Q}(v_{i},\tilde{v}_{i})=\Bigg\{&\mathcal{Q}_{U}
    \in \Delta(\mathcal{U}
    ), \mathbb{E}_{\mathcal{Q}_{U}}[d_i(U,v_{i}) ]=\mathbb{E}_{\mathcal{Q}_{U}}[d_i(U,\tilde{v}_{i})] \Bigg\},  \ \ i \in \{1,2\},\\
    \mathbb{Q}(v_{1},v_{2},\tilde{v}_{1},\tilde{v}_{2})=\Bigg\{&\mathcal{Q}_{U}
    \in \Delta(\mathcal{U}
    ), \mathbb{E}_{\mathcal{Q}_{U}}[(d_e(U,v_{1},v_{2})] =\mathbb{E}_{\mathcal{Q}_{U}}[(d_e(U,\tilde{v}_{1},\tilde{v}_{2})] \Bigg\},
\end{align}
is a closed hyperplane of dimension dim$(\mathbb{Q}(v_{i},\tilde{v}_{i})=$ dim$\mathbb{Q}(v_{1},v_{2},\tilde{v}_{1},\tilde{v}_{2})=|\mathcal{U}|-2.$
Consider the set $B=\Big(\bigcup_{v_1,\tilde{v}_1}\mathbb{Q}(v_{1},\tilde{v}_{1})\Big) \cup \Big(\bigcup_{v_2,\tilde{v}_2}\mathbb{Q}(v_{2},\tilde{v}_{2})\Big) \cup \Big(\bigcup_{v_1,v_2,\tilde{v}_1,\tilde{v}_2}\mathbb{Q}(v_{1},v_2,\tilde{v}_{1},\tilde{v}_2)\Big)$.
The set $B$ is a finite union of hyperplanes of dimension at most $|\mathcal{U}|-2$. Hence, $\Delta(\mathcal{U})\backslash B$ is dense in $\Delta(\mathcal{U})$. If we consider the set $A_0 := ([0,1] \times \Delta(\mathcal{U}))^{|\mathcal{W}_1\times\mathcal{W}_2|}$, 
it follows that the set $A := ([0,1] \times (\Delta(\mathcal{U})\backslash B))^{|\mathcal{W}_1\times\mathcal{W}_2|}$ 
is a dense subset of $A_0$. \\
Let $\Psi: A_0 
\mapsto \Delta(\mathcal{W}_1\times \mathcal{W}_2)^{|\mathcal{U}|}$
a continuous and onto function such that $\Psi((\lambda_{w_1w_2},\mathcal{Q}^{w_1w_2}_U)_{w_1,w_2}) =\frac{\lambda_{w_1w_2}\mathcal{Q}_U^{w_1w_2}}{\mathcal{P}_U}. 
$
Let $\Psi(A)$ denote the image of $A$ under $\Psi$. We show that $\Psi(A)$ is dense in  $\Delta(\mathcal{W}_1\times \mathcal{W}_2)^{|\mathcal{U}|}$. 
Take a distribution $\mathcal{Q}_{W_1W_2|U} \in \Delta(\mathcal{W}_1 \times \mathcal{W}_2)^{|\mathcal{U}|}$. 
Since $A$ is dense in $A_0$, for each distribution $\mathcal{Q}^{w_1w_2}_{U} \in \Delta(U)$, there exists a sequence $(\mathcal{Q}^{w_1w_2}_{U})_{(w_1,w_2)} \in \Delta(U\backslash B)$ that converges to it under the KL-divergence. By the continuity of $\Psi$, the image $\Psi((\mathcal{Q}^{w_1w_2}_{U})_{(w_1,w_2))} \in \Psi(A)$ of $(\mathcal{Q}^{w_1w_2}_{U})_{(w_1,w_2)}$, is a sequence that converges to $\Psi(\mathcal{Q}_{W_1W_2|U}) \in \Delta(\mathcal{W}_1\times \mathcal{W}_2)^{|\mathcal{U}|}$. Therefore, $\Psi(A)$ is dense in $\Delta(\mathcal{W}_1\times \mathcal{W}_2)^{|\mathcal{U}|}$. 

It follows that $\mathbb{Q}^c \cap \mathcal{Q}_0(R_1,R_2)=\mathbb{Q}_{01}\cap\mathcal{Q}_0(R_1,R_2)$ is open and dense in $\mathbb{Q}_0(R_1,R_2)\cap \Delta(\mathcal{W}_1\times\mathcal{W}_2)^{|\mathcal{U}|} = \mathbb{Q}_0(R_1,R_2)$ as desired. \endproof{} \\

\lemma \label{lemmaa7} If $(R_1,R_2)\in\mathbb{R}^2_{+}$, the set $\mathbb{Q}_{02}(R_1,R_2)$ is nonempty, open and dense in $\mathbb{Q}_0(R_1,R_2)$.  \\
\proof{of lemma \ref{lemmaa7}} 
For $(R_1,R_2) \in ]0,+\infty[$, the sets $\mathbb{Q}_{0}(R_1,R_2)$ and $\mathbb{Q}_{02}(R_1,R_2)$ are non-empty. Moreover, the set $\mathbb{Q}_{02}(R_1,R_2)$ is open being defined with strict inequalities on the continuous mutual information function, which means its complement $\mathbb{Q}^c_{02}(R_1,R_2)= \{\mathcal{Q}_{W_1W_2|U} \ s.t. \  R_2 \leq I(U;W_2) \ , \ 
    R_1+R_2 \leq I(U;W_1,W_2)\}$ is closed.
Take a feasible distribution $\mathcal{Q}_{W_1W_2|U} \in \mathbb{Q}_0(R_1,R_2)$ such that  $ R_2 \geq I(U;W_2)$ and $R_1+R_2 \geq I(U;W_1,W_2)$. 
Consider the distributions $\mathcal{P}_{W_1W_2}(w_1,w_2)=\sum_{u}\mathcal{P}(u)\mathcal{Q}(w_1,w_2|u) \forall (w_1,w_2) \in \mathcal{W}_1\times\mathcal{W}_2$ and $\mathcal{P}_{W_2}(w_2)=\sum_{u}\mathcal{P}(u)\mathcal{Q}(w_2|u) \forall w_2 \in \mathcal{W}_2$. For $\varepsilon>0$, consider the perturbed distributions $\mathcal{Q}^{\varepsilon}_{W_1W_2|U} = (1-\varepsilon)\mathcal{Q}_{W_1W_2|U} +\varepsilon\mathcal{P}_{W_1W_2}$, and $\mathcal{Q}^{\varepsilon}_{W_2|U} = (1-\varepsilon)\mathcal{Q}_{W_2|U} +\varepsilon\mathcal{P}_{W_2}$. As $\varepsilon \longrightarrow 0$, we have $\mathcal{Q}^{\varepsilon}_{W_1W_2|U}\longrightarrow \mathcal{Q}_{W_1W_2|U}$, and $\mathcal{Q}^{\varepsilon}_{W_2|U}\longrightarrow \mathcal{Q}_{W_2|U}$. Therefore,
\begin{align}
I_{{\mathcal{Q}^{\varepsilon}}_{W_1W_2|U}}(U;W_1W_2) &\leq (1-\varepsilon)\cdot I_{{\mathcal{Q}}_{W_1W_2|U}}(U;W_1W_2) + \varepsilon \cdot I_{\mathcal{P}_{W_1W_2}}(U;W_1W_2) \label{casesz01}\\ 
&< I_{{\mathcal{Q}}_{W_1W_2|U}}(U;W_1W_2) \label{casesz02}\\ 
&\leq R_1+R_2. \label{casesz03} \end{align}
Similarly,
\begin{align}
I_{{\mathcal{Q}^{\varepsilon}}_{W_2|U}}(U;W_2) &\leq (1-\varepsilon)\cdot I_{{\mathcal{Q}}_{W_2|U}}(U;W_2) + \varepsilon \cdot I_{\mathcal{P}_{W_2}}(U;W_2) \label{casesz1}\\ 
&< I_{{\mathcal{Q}}_{W_2|U}}(U;W_2) \label{casesz2}\\ 
&\leq R_2.
\label{casesz3}
\end{align}
Equations \eqref{casesz01} and \eqref{casesz1} follow from the convexity of the mutual information with respect to $\mathcal{Q}_{W_1W_2|U}$ and $\mathcal{Q}_{W_2|U}$ respectively for fixed $\mathcal{P}_U$. 
The strict inequalities in \eqref{casesz02} and \eqref{casesz2} follow since $I_{\mathcal{P}_{W_1W_2}}(U;W_1W_2) =0$ and $I_{\mathcal{P}_{W_2}}(U;W_2) =0$ and $\varepsilon >0$, and last inequalities in equations \eqref{casesz03} and \eqref{casesz3} come from the definition of the set $\mathbb{Q}_0(R_1,R_2)$. This means that both distributions $\mathcal{Q}^{\varepsilon}_{W_1W_2|U}$ and $\mathcal{Q}^{\varepsilon}_{W_2|U}$ belong to the set $\mathbb{Q}_0(R_1,R_2)$. Hence, the set $\mathbb{Q}_{02}(R_1,R_2)$ is dense in $\mathbb{Q}_0(R_1,R_2)$ which concludes the proof of lemma \ref{lemmaa7}. \endproof{} \\
Since $\mathbb{Q}_{01}$ and $\mathbb{Q}_{02}(R_1,R_2)$ are open and dense, $\mathbb{Q}_{01} \cap \mathbb{Q}_{02}(R_1,R_2)$ is also open and dense in $\mathbb{Q}_{0}(R_1,R_2)$. We now show that 
$D_e^{\star}(R_1,R_2) = \tilde{D}_e(R_1,R_2)$. In fact, the function $$\mathcal{Q}_{W_1W_2|U} \mapsto \underset{\mathcal{Q}_{V_1|W_1W_2}\in \mathbb{Q}_{1}(\mathcal{Q}_{W_1W_2|U}) \atop \mathcal{Q}_{V_2|W_2}\in \mathbb{Q}_{2}(\mathcal{Q}_{W_2|U})}{\max}\mathbb{E}_{\mathcal{P}_{U} \mathcal{Q}_{W_1W_2|U}\atop \mathcal{Q}_{V_1|W_1W_2}\mathcal{Q}_{V_2|W_2}}\Bigg[
    d_e(U,V_1,V_2)\Bigg]$$ is upper semi-continuous (u.s.c) and the infimum of an u.s.c function over a dense set is the infimum over the full set. \\ 
    In this part of the proof, the assumption that each decoder chooses the optimal action that is worst for the encoder plays an important role. In fact, if decoders were to choose the pair of actions that is best for the encoder's distortion, our function becomes $$\mathcal{Q}_{W_1W_2|U} \mapsto \underset{\mathcal{Q}_{V_1|W_1W_2}\in \mathbb{Q}_{1}(\mathcal{Q}_{W_1W_2|U}) \atop \mathcal{Q}_{V_2|W_2}\in \mathbb{Q}_{2}(\mathcal{Q}_{W_2|U})}{\min}\mathbb{E}_{\mathcal{P}_{U} \mathcal{Q}_{W_1W_2|U}\atop \mathcal{Q}_{V_1|W_1W_2}\mathcal{Q}_{V_2|W_2}}\Bigg[
    d_e(U,V_1,V_2)\Bigg]$$ which is lower semi continuous. The infimum of a lower semi continuous (l.s.c) function over a dense subset  $\mathbb{Q}_{01} \cap \mathbb{Q}_{02}(R_1,R_2)$ might be greater than the infimum over the whole set $\mathbb{Q}_{0}(R_1,R_2)$. However, this is only the case whenever the information is constrained, and the information constraint is binding at optimum and all posterior beliefs of each decoder induce actions between which decoder is indifferent. This case in nongeneric in our class of persuasion games: if we slightly perturb the distortion functions of our decoders, we perturb the points of indifference for each decoder, and thus the points of discontinuity in our l.s.c. or u.s.c. This ends the proof of lemma \ref{lemmmaa5}. \endproof{}

\subsubsection{Controlling Distortions }
\definition \label{belpp} Fix $(R_1,R_2) \in \mathbb{R}^2_{+}$, $n\in\mathbb{N}$, a triplet $(\sigma,\tau_1,\tau_2) \in \mathcal{S}(n,R_1,R_2)$, $t\in\{1,...,n\}$ and a message pair $(m_1,m_2) \in \{1,2,..2^{\lfloor nR_1 \rfloor}\}\times\{1,2,..2^{\lfloor nR_2\rfloor}\}$. We denote by $\mathcal{P}^{m_1,m_2}_{t} \in \Delta(\mathcal{U})$, and $\mathcal{P}^{m_2}_{t} \in \Delta(\mathcal{U})$ the beliefs on $u_t$ conditional to $m_1,m_2$ and $m_2$ respectively defined as follows: 
\begin{align}
    \mathcal{P}^{m_1,m_2}_{t}(u) =& \mathcal{P}(U_t=u|M_1=m_1,M_2=m_2) \ \forall u\in \mathcal{U}, \\
    \mathcal{P}^{m_2}_{t}(u) =& \mathcal{P}(U_t=u|M_2=m_2) \  \forall u\in \mathcal{U}.
\end{align}

\definition Let $\Tilde{A}_t(\mathcal{P}^{\sigma}_{M_1M_2|U^n},m_1,m_2)$ the set of action pairs $(v_1,v_2)$ that are optimal for the decoders but worst for the encoder for respective beliefs $\mathcal{P}^{m_1,m_2}_{t}(\cdot|m_1m_2)$ and $\mathcal{P}^{m_2}_{t}(\cdot|m_2)$:
\begin{align}
    \Tilde{A}_t(\mathcal{P}^{\sigma}_{M_1M_2|U^n},m_1,m_2)= \underset{v_1 \in V_1^{\star}(\mathcal{P}^{m_1,m_2}_{t}(\cdot|m_1m_2)) \atop v_2 \in V_2^{\star}(\mathcal{P}^{m_2}_{t}(\cdot|m_2))}{\argmax}\Big\{\sum_{u}\mathcal{P}^{m_1,m_2}_{t}(u)d_e(u,v_1,v_2)\Big\}
\end{align}

 \definition Fix $(R_1,R_2) \in \mathbb{R}^2_{+}$, $n\in\mathbb{N}$, a triplet $(\sigma,\tau_1,\tau_2) \in \mathcal{S}(n,R_1,R_2)$, $t \in \{1,...,n\}$ and a message pair $(m_1,m_2) \in \{1,2,..2^{\lfloor nR_1 \rfloor}\}\times\{1,2,..2^{\lfloor nR_2\rfloor}\}$. For a sequence $(m_1,m_2,w_1^n,w_2^n)$, and $\alpha>0$, we define the set of indices for which posterior belief $\mathcal{P}^{m_1,m_2}_{t}$ given in Definition \ref{belpp}, and theoretical belief $\mathcal{Q}_{U|W_{1}W_{2}}$ given in Definition \ref{posbel18}, are close as follows:
 \begin{align}
     T_{\alpha}(m_1,m_2,w_1^n,w_2^n)=\bigg\{ t \in \{1,...,n\} \ : \ \max\big(D(\mathcal{P}^{m_1,m_2}_{t} ||\mathcal{Q}_U^{w_{1,t},w_{2,t}}
     ) 
     , D(\mathcal{P}^{m_2}_{t} ||\mathcal{Q}_U^{w_{2,t}}
     )\big)\leq \frac{\alpha^2}{2\ln2}\bigg\}.
 \end{align}
 \definition For a sequence $(w_1^n,w_2^n)$ and a pair $(w_1,w_2) \in \mathcal{W}_1 \times \mathcal{W}_2$, the empirical frequency of $(w_1,w_2)$ in $(w_1^n,w_2^n)$ is given by:
 \begin{align}
     \mathrm{freq}_{w_1,w_2}(w_1^n,w_2^n)=\frac{1}{n}\bigg|t=\{1,...,n\} \ : \ (w_{1,t},w_{2,t})=(w_1,w_2)\bigg|.
 \end{align}
 For $\alpha,\gamma,\delta >0$, let 
 \begin{align}
  B_{\alpha,\gamma,\delta}=\bigg\{ (m_1,m_2,w_1^n,w_2^n) : \  &\frac{|T_{\alpha}(m_1,m_2,w_1^n,w_2^n)|}{n} \geq 1-\gamma,\  \nonumber \\ &\sum_{(w_1,w_2)}|\mathcal{P}(w_1,w_2) -  \mathrm{freq}_{w_1,w_2}(w_1^n,w_2^n)| \leq \delta \bigg\}. \end{align}
where $\forall w_1,w_2 \ \mathcal{P}(w_1,w_2)=\sum_{u}\mathcal{P}(u)\mathcal{Q}(w_1,w_2|u)$.

\definition 
Let $n \in \mathbb{N}^{\star}$, and $(R_1,R_2) \in \mathbb{R}^2_{+}$. Given a strategy $\sigma$ of the encoder, the induced expected distortion is given as follows:
\begin{align}
    D^{\sigma}_e(R_1,R_2)=\max_{\tau_1\in BR_1(\sigma) \atop \tau_2\in BR_2(\sigma)} \mathbb{E}_{\mathcal{P}^{\sigma,\tau_1,\tau_2}}[d^n_e(\sigma,\tau_1,\tau_2)].
\end{align}
 \definition Given an encoding strategy $\sigma$ and a pair of messages $(m_1,m_2)$, the encoder's expected distortion is given as follows:
 \begin{align}
     D^{t}_{e}(\mathcal{P}^{\sigma}_{M_1M_2|U^n},m_1,m_2)=\underset{v_1 \in V_1^{\star}(\mathcal{P}^{m_1,m_2}_{t}) \atop v_2 \in V_2^{\star}(\mathcal{P}^{m_2}_{t})}{\max}\sum_{u}\mathcal{P}^{m_1m_2}_t(u)d_e(u,v_{1},v_{2}).
 \end{align}
 
 \definition We denote by $D^w_{e}(\mathcal{Q}_{W_1W_2|U},w_1,w_2)$ the encoder's distortion defined as a function of the beliefs of the decoders as follows:
\begin{align}
    D^w_{e}(\mathcal{Q}_{W_1W_2|U},w_1,w_2)=\underset{(v_1,v_2) \in V_1^{\star}({Q}^{w_1,w_2}_U) \times \atop V_2^{\star}({Q}^{w_2}_U)}{\max}\sum_{u}\mathcal{Q}_U^{w_1,w_2}(u)d_e(u,v_1,v_2).
\end{align}
 \lemma \label{lem9} Given $(n,R_1,R_2)$, for all $\sigma_{M_1M_2W_1^nW_2^n|U^n}$ we have, \begin{align}
     | D^{\sigma}_e(R_1,R_2)- \tilde{D}_e(R_1,R_2)| \leq (\alpha+2\gamma+\delta)||D|| + (1-\mathcal{P}^{\sigma}(B_{\alpha,\gamma,\delta}))||D||.
 \end{align}
 where $||D|| =\max_{u,v_1,v_2}|d_e(u,v_1,v_2)|$ is the greatest absolute value of the encoder's distortion. \\
 \proof{(of lemma \ref{lem9})} The strategy $\sigma$ induces a joint probability distribution $\mathcal{P}^{\sigma}$ over $\mathcal{U}^n \times\{1,2,..2^{\lfloor nR_1 \rfloor}\}\times\{1,2,..2^{\lfloor nR_2\rfloor}\} \times \mathcal{W}^n_1\times\mathcal{W}^n_2$ such that for all $u^n,m_1,m_2,w_1^n,w_2^n$,
 \begin{align}
     \mathcal{P}^{\sigma}(u^n,m_1,m_2,w_1^n,w_2^n)=  \prod_{t=1}^n\mathcal{P}(u^n)\mathcal{P}^{\sigma}(m_1,m_2,w_1^n,w_2^n|u^n).
 \end{align}
     Let $\mathcal{P}^{\sigma}_{W_1^nW_2^n}$ the marginal distribution of $\mathcal{P}^{\sigma}$ over $\mathcal{W}_1^n\times\mathcal{W}_2^n$. For each $t$, and for  each pair $(w_{1,t},w_{2,t})$, decoder $\mathcal{D}_1$ chooses an optimal action $v_{1,t} \in V^{\star}_1(\mathcal{Q}^{w_{1,t},w_{2,t}}_U)$, and decoder $\mathcal{D}_2$ chooses an optimal action $v_{2,t} \in V^{\star}_2(\mathcal{Q}^{w_{2,t}}_U)$. 
     If the action pair $(v_{1,t},v_{2,t})$ belongs to $\tilde{A}(\mathcal{Q}_{W_1W_2|U},w_{1,t},w_{2,t})$, then it's the worst pair for the encoder. It follows that \begin{align}
         D^{\sigma}_e(R_1,R_2)=\sum_{m_1,m_2}\mathcal{P}^{\sigma}(m_1,m_2)\frac{1}{n}\sum_{t=1}^n D^t_{e}(\mathcal{P}^{\sigma}_{M_1M_2|U^n},m_1,m_2).
     \end{align}
 Since the set of belief pairs such that $ |\tilde{A}(\mathcal{Q}_{W_1W_2|U},w_{1,t},w_{2,t})
 |=1$ is open, there exists $\alpha_0>0$, such that for all $m_1,m_2$ and for all $t$, we have:   
 \begin{align}
    \max\bigg\{ D(\mathcal{P}^{m_1,m_2}_{t} || \mathcal{Q}^{w_{1,t},w_{2,t}}_U),D(\mathcal{P}^{m_2}_{t} || \mathcal{Q}^{w_{2,t}}_U)\bigg\} \leq \alpha_0 \implies  \Tilde{A}(\mathcal{Q}_{W_1W_2|U},w_{1,t},w_{2,t})=\Tilde{A}_t(\mathcal{P}^{\sigma}_{M_1M_2|U^n},m_1,m_2). \label{atildez}
 \end{align}
  Whenever $\tilde{A}(\mathcal{Q}_{W_1W_2|U},w_{1,t},w_{2,t})
 $ is a singleton, denote $(v_1(\mathcal{Q}^{w_{1,t},w_{2,t}}_U),v_2(\mathcal{Q}^{w_{2,t}}_U))$ the unique (worst) optimal action pair for the encoder's distortion.  
 From now on, we assume that $\alpha \in (0,\alpha_0)$.
Equation \eqref{atildez} implies that 
for each $t \in  T_{\alpha}(m_1,m_2,w_1^n,w_2^n)$, the action pair chosen by the decoders for problem t is $(v_1(\mathcal{Q}^{w_{1,t},w_{2,t}}_U),v_2(\mathcal{Q}^{w_{2,t}}_U))$. This means that the set $T_{\alpha}(m_1,m_2,w_1^n,w_2^n)$ is the set of indices $t$ 
for which the information transmission is successful. 
\lemma \label{lem10} Let $(R_1,R_2) \in \mathbb{R}^2_{+}$. For each $(m_1,m_2,w_1^n,w_2^n) \in B_{\alpha,\gamma,\delta}$, 
\begin{align}
    \bigg|  \frac{1}{n}\sum^n_{t=1} D^{t}_{e}(\mathcal{P}^{\sigma}_{M_1M_2|U^n},m_1,m_2)- \tilde{D}_e(R_1,R_2)\bigg| \leq (\alpha+2\gamma+\delta)||D||
\end{align}
 where $||D|| =\max_{u,v_1,v_2}|d_e(u,v_1,v_2)|$ is the greatest absolute value of the encoder's distortion. \\
\proof{(of lemma \ref{lem10})}
 We have :
\begin{align}
    \bigg|  \frac{1}{n}\sum^n_{t=1} D^{t}_{e}(\mathcal{P}^{\sigma}_{M_1M_2|U^n},m_1,m_2)- \tilde{D}_e(R_1,R_2)\bigg| 
    \leq& \bigg|  \frac{1}{n}\sum^n_{t\in T_{\alpha}(m_1,m_2,w_1^n,w_2^n)} D^{t}_{e}(\mathcal{P}^{\sigma}_{M_1M_2|U^n},m_1,m_2)- \tilde{D}_e(R_1,R_2)\bigg| \nonumber \\ +& \bigg|  \frac{1}{n}\sum^n_{t\notin T_{\alpha}(m_1,m_2,w_1^n,w_2^n)} D^{t}_{e}(\mathcal{P}^{\sigma}_{M_1M_2|U^n},m_1,m_2)- \tilde{D}_e(R_1,R_2)\bigg| \\
    \leq& \bigg|  \frac{1}{n}\sum^n_{t\in T_{\alpha}(m_1,m_2,w_1^n,w_2^n)} D^{t}_{e}(\mathcal{P}^{\sigma}_{M_1M_2|U^n},m_1,m_2)- \tilde{D}_e(R_1,R_2)\bigg| + \gamma||D||.
\end{align}
Then,
\begin{align}
    &\bigg|  \frac{1}{n}\sum^n_{t\in T_{\alpha}(m_1,m_2,w_1^n,w_2^n)} D^{t}_{e}(\mathcal{P}^{\sigma}_{M_1M_2|U^n},m_1,m_2)- \tilde{D}_e(R_1,R_2)|\bigg| \\ \leq &\bigg|  \frac{1}{n}\sum^n_{t\in T_{\alpha}(m_1,m_2,w_1^n,w_2^n)} \big[D^{t}_{e}(\mathcal{P}^{\sigma}_{M_1M_2|U^n},m_1,m_2)- D^{w}_{e}(\mathcal{Q}_{W_1W_2|U},w_{1,t},w_{2,t})\big]\bigg| \nonumber \\ + &\bigg|  \frac{1}{n}\sum^n_{t\in T_{\alpha}(m_1,m_2,w_1^n,w_2^n)}\big[D^w_{e}(\mathcal{Q}_{W_1W_2|U},w_{1,t},w_{2,t})- \tilde{D}_e(R_1,R_2)\big]\bigg| 
\end{align}
Since $\alpha\leq \alpha_0$, for each $t\in T_{\alpha}(m_1,m_2,w_1^n,w_2^n)$, $\Tilde{A}(\mathcal{Q}_{W_1W_2|U},w_{1,t},w_{2,t})=\Tilde{A}_t(\mathcal{P}_{M_1M_2|U},m_1,m_2)$, therefore, 
\begin{align}
    \bigg| D^{t}_{e}(\mathcal{P}^{\sigma}_{M_1M_2|U^n},m_1,m_2)- D^w_{e}(\mathcal{Q}_{W_1W_2|U},w_{1,t},w_{2,t})\bigg| \leq& \sum_u \bigg| \mathcal{P}_t^{m_1,m_2}(u) - \mathcal{Q}^{w_1w_2}_U(u)\bigg|\cdot ||D|| \\ \leq& \big|\big|\mathcal{P}_t^{m_1,m_2} - \mathcal{Q}^{w_1w_2}_U\big|\big| \cdot \big|\big|D\big|\big| \leq \alpha \big|\big|D\big|\big|,
\end{align} where the second inequality comes from Pinsker's inequality: $||p-q||\leq\sqrt{2\ln2D(p||q)}$ and the definition of $\mathcal{T}_{\alpha}(m_1,m_2,w_1^n,w_2^n)$. It follows:
  \begin{align}
    &\bigg|  \frac{1}{n}\sum^n_{t\in T_{\alpha}(m_1,m_2,w_1^n,w_2^n)} D^t_{e}(\mathcal{P}^{\sigma}_{M_1M_2|U^n},m_1,m_2)- \tilde{D}_e(R_1,R_2)|\bigg| \\ \leq  &\alpha ||D|| + \bigg|  \frac{1}{n}\sum^n_{t\in T_{\alpha}(m_1,m_2,w_1^n,w_2^n)}D^w_{e}(\mathcal{Q}_{W_1W_2|U},w_{1,t},w_{2,t})- \tilde{D}_e(R_1,R_2)\bigg|. 
\end{align}
Now from $\frac{|T_{\alpha}(m_1,m_2,w_1^n,w_2^n)|}{n} \geq 1-\gamma$, we have:
\begin{align}
    \bigg|  \frac{1}{n}\sum_{t\in T_{\alpha}(m_1,m_2,w_1^n,w_2^n)}D^w_{e}(\mathcal{Q}_{W_1W_2|U},w_{1,t},w_{2,t})- \tilde{D}_e(R_1,R_2)\bigg| \leq \bigg|  \frac{1}{n}\sum^n_{t=1}D^w_{e}(\mathcal{Q}_{W_1W_2|U},w_{1,t},w_{2,t})- \tilde{D}_e(R_1,R_2)\bigg| +\gamma||D||.
\end{align}
 We have $\forall w_1,w_2, \ \mathcal{P}(w_1,w_2)=\sum_{u}\mathcal{P}(u)\mathcal{Q}(w_1,w_2|u)$, then
\begin{align}
    \bigg|  \frac{1}{n}\sum^n_{t=1}D^w_{e}(\mathcal{Q}_{W_1W_2|U},w_{1,t},w_{2,t})- \tilde{D}_e(R_1,R_2)\bigg| 
    &\leq \sum_{w_1,w_2}\bigg|(\textrm{freq}_{w_1,w_2}(w_1^n,w_2^n)-\mathcal{P}(w_1w_2))\bigg|\cdot  ||D||  \\ &\leq \delta||D||.
\end{align}

\endproof{}

\subsubsection{Special Cases}
We begin by investigating some particular cases where at least one of the rates equals zero. Then we will prove our result for the general case and we control the beliefs of our decoders. \\
\remark If both decoders have the same distrotion functions, then they can be considered as one, and the persuasion game will be reduced to the point-to-point case as in \cite{jet}.

Let $Q^n_{\sigma,\tau_1,\tau_2}$ denote a distribution over $\mathcal{U} \times \mathcal{V}_1 \times \mathcal{V}_2$ that averages the probability of occurrence of $(u,v_1,v_2)$ in a triplet of sequences $(U^n,V_1^n,V_2^n)$ with respect to coding pair $\sigma,\tau_1,\tau_2$ defined as follows: 
\begin{align}
        Q^n_{\sigma,\tau_1,\tau_2}(u,v_1,v_2)=\frac{1}{n} \sum_{t=1}^n \mathcal{P}\{(U_t,V_{1,t},V_{2,t})=(u,v_1,v_2)\}, \ \ \ \forall (u,v_1,v_2).
    \end{align}
      Since the source is memory-less we have
${Q^n_{\sigma,\tau_1,\tau_2}}_{U}=\mathcal{P}_U.$ 
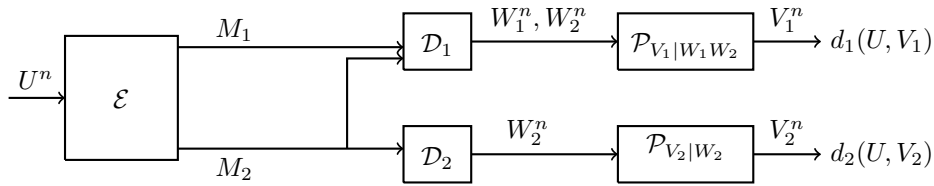
\begin{figure}[h!]
    \centering
    \begin{tikzpicture}[xscale=3,yscale=1.5]
\draw[thick,->](0.5,0.25)--(0.75,0.25);
\draw[thick,-](0.75,-0.3)--(1.25,-0.3)--(1.25,0.8)--(0.75,0.8)--(0.75,-0.3);
\node[below, black] at (1,0.4) {$\mathcal{E}$};
\draw[thick, ->](1.25,0.7)--(2.25,0.7);
\draw[thick, ->](2,-0.2)--(2,0.6)--(2.25,0.6);
\draw[thick, ->](1.25,-0.2)--(2.25,-0.2);
\draw[thick,-](2.25,0.5)--(2.55,0.5)--(2.55,1)--(2.25,1)--(2.25,0.5);
\node[below, black] at (2.4,0.9) { $\mathcal{D}_{1}$};
\draw[thick,-](2.25,-0.5)--(2.55,-0.5)--(2.55,0)--(2.25,0)--(2.25,-0.5);
\node[below, black] at (2.4,-0.1) { $\mathcal{D}_{2}$};
\node[above, black] at (0.62,0.25) {$U^n$};
\node[above, black] at (1.5,0.67){$M_1$};
\node[below, black] at (1.5,-0.2){$M_2$};
  \draw[thick,-](3.2,-0.5)--(3.8,-0.5)--(3.8,0)--(3.2,0)--(3.2,-0.5);
\draw[thick,-](3.2,0.5)--(3.8,0.5)--(3.8,1)--(3.2,1)--(3.2,0.5);
\draw[thick,->](2.55,0.75)--(3.2,0.75);
\node[above, black] at (3.95,0.75) {$V_1^n$};
\node[above, black] at (3.95,-0.25) {$V_2^n$};
\draw[thick,->](3.8,0.75)--(4.1,0.75);
\draw[thick,->](3.8,-0.25)--(4.1,-0.25);
\node[above, black] at (3.5,0.5) {$\mathcal{P}_{V_1|W_1W_2}$};
\node[above, black] at (3.5,-0.4) {$\mathcal{P}_{V_2|W_2}$};
\draw[thick,->](2.55,-0.25)--(3.2,-0.25);
\node[above, black] at (2.85,0.75) {$W_1^n,W_2^n$};
\node[above, black] at (2.8,-0.25) {$W_2^n$};
\node[right, black] at (4.1,-0.25) {$d_2(U,V_2)$};
\node[right, black] at (4.1,0.75) {$d_1(U,V_1)$};
    \end{tikzpicture}
    \caption{Achievability of Successive Refinement Source Coding Setup.}
    \label{fig:successive}
\end{figure}

\subsubsection{$R_1=R_2=0$} 
Assume the prior belief $\mathcal{P}_U$ is fixed and shared by both decoders at the beginning of the game. 
Since $R_1=R_2=0$, message sets are singletons, and the only possible encoding strategy $\sigma_0$ is given by $\sigma_0:\mathcal{U}^n\longrightarrow\{1\}\times\{1\}$. The codebook consists of two sequences $W_2^n(1)$ and $W_1^n(1,1)$ only. Let $(v_{1,0},v_{2,0})$ denote the action pair that corresponds to the decoders' prior $\mathcal{P}_U$ and maximizes the encoder's long run distortion. This action pair will be played at each repetition of the game, i.e $v^n_{1,0}=(v_{1,0},v_{1,0},...,v_{1,0})$ and $v^n_{2,0}=(v_{2,0},v_{2,0},...,v_{2,0})$. The corresponding pair of decoding strategies is denoted by $(\tau_{1,0},\tau_{2,0}) \in BR_1(\sigma_0)\times BR_2(\sigma_0)$.
The set of target distributions is given by $\mathbb{Q}_0(0,0)=\{\mathcal{Q}_{W_1W_2|U} \ \mathrm{ s.t. } \ I(U;W_2)=I(U;W_1,W_2)=0\}$. This means that random variables $W_1$ and $W_2$ are independent from $U$, i.e $\mathcal{Q}_{W_1W_2|U}=\mathcal{Q}_{W_1W_2}$ and no information can be communicated to the decoders. 
Therefore, the following result holds:
\lemma  \label{zerorates}  $D_e^{\star}(0,0)=D_e^n(0,0) \ \ \forall n \in \mathbb{N}^{\star}$. \\
\proof
\begin{align}
D_e^n(0,0) =& \underset{\sigma_0}{\inf}\underset{\tau_{1,0} \in BR_{d_1}(\sigma_0), \atop \tau_{2,0} \in BR_{d_2}(\sigma_0) }{\max}d_e^n(\sigma,\tau_{1},\tau_{2}) \\
=& d_e^n(\sigma_0,\tau_{1,0},\tau_{2,0})\\
        =& \sum_{u^n,v_{1,0}^n,v_{2,0}^n}\mathcal{P}^{\sigma,\tau_1,\tau_2}(u^n,v_{1,0}^n,v_{2,0}^n)\frac{1}{n}\sum_{t=1}^n d_e(u_t,v_{1,0},v_{2,0}) \\
        =& \frac{1}{n}\sum_{t=1}^n\sum_{u_t,v_{1,0},v_{2,0}}\mathcal{P}^{\sigma_0,\tau_{1,0},\tau_{2,0}}(u_t,v_{1,0},v_{2,0}) d_e(u_t,v_{1,0},v_{2,0}) \\
        =& \sum_{u,v_{1,0},v_{2,0}}\sum_{t=1}^n\frac{1}{n}\mathcal{P}^{\sigma_0,\tau_{1,0},\tau_{2,0}}(U_t=u,V_{1,t}=v_{1,0},V_{2,t}=v_{2,0}) d_e(u,v_{1,0},v_{2,0}) \\
        =&\sum_{u,v_{1,0},v_{2,0}}{Q^n_{\sigma_0,\tau_{1,0},\tau_{2,0}}}_{UV_1V_2}(u,v_{1,0},v_{2,0})d_e(u,v_{1,0},v_{2,0}) \\
        =&  \mathbb{E}_{\mathcal{P}_U
        }[d_e(U,v_{1,0},v_{2,0})] \\
       =&\underset{\mathcal{Q}_{W_1W_2}\in\mathbb{Q}_0(0,0)}{\inf}\underset{\mathcal{Q}_{V_1|W_1W_2}\in \mathbb{Q}_{1}(\mathcal{Q}_{W_1W_2}) \atop \mathcal{Q}_{V_2|W_2}\in \mathbb{Q}_{2}(\mathcal{Q}_{W_2})}{\max}\mathbb{E}_{\mathcal{P}_{U} \mathcal{Q}_{W_1W_2}\atop \mathcal{Q}_{V_1|W_1W_2}\mathcal{Q}_{V_2|W_2}}\Big[
    d_e(U,V_1,V_2)\Big]\\
        =&D^{\star}_e(0,0).
    \end{align}
 \endproof{}
 
 \subsubsection{$R_1>0 \And R_2=0$}

 Random variables $W_2$ and $U$ are independent for $R_1>0$ and $R_2=0$, i.e. $\mathcal{Q}_{W_1W_2|U}=\mathcal{Q}_{W_2}\mathcal{Q}_{W_1|W_2U}$. This means that decoder $\mathcal{D}_2$ will repeatedly chose the action $v_{2,0} \in V^{\star}(\mathcal{P}_U)$ that corresponds to its prior belief $\mathcal{P}_U$ and maximizes the encoder's distortion. The persuasion game is thus reduced to the point-to-point problem with one decoder $\mathcal{D}_1$. 
 In that case, the coding problem to be solved by the encoder is as follows:
\begin{align}
 D_e^n=\underset{\sigma}{\inf}\underset{\tau_{1} \in BR_{d_1}(\sigma)}{\max}d_e^n(\sigma,\tau_1)
\end{align}
where $BR_{d_1}(\sigma)=\argmin_{\tau_1}d_1^n(\sigma,\tau_1)$.  This problem has been investigated in point to point and JET.
The set of target distributions is given as follows $\mathbb{Q}_0(R_1,0)=\{\mathcal{Q}_{W_1|U} \ \mathrm{ s.t. } \ R_1 \geq I(U;W_1)\}$.
Given $\mathcal{Q}_{UW_1}$, the set of single-letter best responses of decoder $\mathcal{D}_1$ is given by
$$\mathbb{Q}_1(\mathcal{Q}_{W_1|U})=\argmin_{\mathcal{Q}_{V_1|W_1}}\mathbb{E}\big[d_1(U,V_1)\big].$$

\definition \label{psieq}We denote by $\Psi_e(q)$ the encoder's expected distortion for belief $q \in \Delta(\mathcal{U})$ i.e,
$$\Psi_e(q)=\underset{(v_1,v_2) \in V_1^{\star}(q) \times \atop V_2^{\star}(q)}{\max}\mathbb{E}_q\bigg[d_e(U,v_1,v_2)\bigg],$$ 

\definition A family of pairs  $(\lambda_{w_1},\mathcal{Q}^{w_1}_U)_{w_1} \in ([0,1] \times \Delta(\mathcal{U}))^{|\mathcal{W}_1|}$ is a splitting for decoder $\mathcal{D}_1$ if
\begin{align}
\sum_{w_1}\lambda_{w_1}&=1, \label{splittingr1r20} \\
    \sum_{w_1}\lambda_{w_1}\mathcal{Q}^{w_1}_U &=\mathcal{P}_U.  \label{splittingr1r202} 
\end{align}
For every $w_1 \in \mathcal{W}_1$, the weight $\lambda_{w_1}$ is given by $\lambda_{w_1}=\mathcal{P}(w_1)=\sum_u\mathcal{P}(u)\mathcal{Q}(w_1|u)$.  
The encoder's optimal distortion can be reformulated as a convexification of its expected distortion as follows: 
\begin{align}
 D^{\star}_e(R_1,0)=&\underset{\mathcal{Q}_{W_1|U}\in\mathbb{Q}_0(R_1,0)}{\inf}\underset{\mathcal{Q}_{V_1|W_1}\in \mathbb{Q}_{1}(\mathcal{Q}_{W_1W_2|U})}{\max}\mathbb{E}_{\mathcal{P}_{U} \mathcal{Q}_{W_1|U}\atop \mathcal{Q}_{V_1|W_1}}\Bigg[
   d_e(U,V_1)\Bigg]. \label{optdistor1r20} \nonumber \\
   =& \underset{(\lambda_{w_1},\mathcal{Q}_U^{w_1})_{w_1}}{\inf}\sum_{w_1}\lambda_{w_1}\Psi_e(\mathcal{Q}_U^{w_1}).
\end{align}
\remark The auxiliary random variable $W_1 \in \mathcal{W}_1$ satisfies $|\mathcal{W}_1|=\min\{|\mathcal{U}|+1,|\mathcal{V}_1|\}$. 
\theorem{Encoder Commitment, theorem 3.1 in \cite{jet}}
\label{mainresultr1r20}
\begin{align} 
\forall \ \varepsilon>0, \ \exists \hat{n} \in \mathbb{N}, \  \forall n \geq \hat{n}, \ \ D^n_e(R_1,0) &\leq D_e^{\star}(R_1,0) + \varepsilon. \hspace{1cm}\mathrm{(achievability) }  \\
 \hspace{3.13cm} \forall n \in \mathbb{N}, \ \  D_e^n(R_1,0) &\geq D_e^{\star}(R_1,0). \hspace{1.58cm}\mathrm{(converse) } 
\end{align}

\subsubsection{$R_1=0 \And R_2>0$ } 

If $R_1=0$ and $R_2>0$, random variables $W_1$ and $U$ are independent. Hence, the encoder can transmit information to decoder $\mathcal{D}_{2}$, to which decoder $\mathcal{D}_{1}$ has access. Therefore, both decoders will have the same posterior belief $\mathcal{Q}^{w_2}_U \in \Delta(\mathcal{U}) \  \forall w_2 \in \mathcal{W}_2$. Actions $V^n_1$ and $V^n_2$ are drawn according to $\mathcal{Q}_{V^n_1|W^n_2}$ and $\mathcal{Q}_{V^n_2|W^n_2}$ respectively. If the objectives of both decoders are aligned, then the persuasion game can be reduced to one decoder as in \cite{jet}. 
Otherwise, the persuasion game is an extension to the problem investigated in \cite{jet} with two decoders that observe the same information from the encoder and hence have the same belief $q \in \Delta(\mathcal{U})$. 

In that case, the set of target distributions is defined as follows: $\mathbb{Q}_0(0,R_2)=\{\mathcal{Q}_{W_2|U} \ \mathrm{ s.t. } \ R_2 \geq I(U;W_2)\}$.


We consider an auxiliary random variable $W_2 \in \mathcal{W}_2$ with $|\mathcal{W}_2|=\min\{|\mathcal{U}|+1,|\mathcal{V}_1|,|\mathcal{V}_2| \}$. 
The set of target distributions is given as follows $\mathbb{Q}_0(R_1,0)=\{\mathcal{Q}_{W_2|U} \ \mathrm{ s.t. } \ R_2 \geq I(U;W_2)\}$.
Given $\mathcal{Q}_{UW_2}$, the set of single-letter best responses of decoders $\mathcal{D}_1$ and $\mathcal{D}_2$ are given by
$$\mathbb{Q}_1(\mathcal{Q}_{W_2|U})=\argmin_{\mathcal{Q}_{V_1|W_2}}\mathbb{E}\big[d_1(U,V_1)\big].$$
$$\mathbb{Q}_2(\mathcal{Q}_{W_2|U})=\argmin_{\mathcal{Q}_{V_2|W_2}}\mathbb{E}\big[d_2(U,V_2)\big].$$

\definition A family of pairs  $(\lambda_{w_2},\mathcal{Q}^{w_2}_U)_{w_2} \in ([0,1] \times \Delta(\mathcal{U}))^{|\mathcal{W}_2|}$ is a splitting if
\begin{align}
\sum_{w_2}\lambda_{w_2}&=1, \label{splittingr2r101} \\
    \sum_{w_2}\lambda_{w_2}\mathcal{Q}^{w_2}_U &=\mathcal{P}_U.  \label{splittingr2r102} 
\end{align}
For $w_2 \in \mathcal{W}_2$, the weights  $\lambda_{w_2}$ are given by $\mathcal{P}(w_2)=\sum_u\mathcal{P}(u)\mathcal{Q}(w_2|u)$.  
The encoder's optimal distortion can be reformulated as a convexification of its expected distortion as follows:
\begin{align}
 D^{\star}_e(0,R_2)=&\underset{\mathcal{Q}_{W_2|U}\in\mathbb{Q}_0(0,R_2)}{\inf}\underset{\mathcal{Q}_{V_1|W_2}\in \mathbb{Q}_{1}(\mathcal{Q}_{W_2|U}) \atop \mathcal{Q}_{V_2|W_2}\in \mathbb{Q}_{2}(\mathcal{Q}_{W_2|U})}{\max}\mathbb{E}_{\mathcal{P}_{U} \mathcal{Q}_{W_2|U}\atop \mathcal{Q}_{V_1|W_2}  \mathcal{Q}_{V_2|W_2}}\Bigg[
    d_e(U,V_1,V_2)\Bigg] \label{optdistor2r10} \nonumber \\ 
    =& \underset{(\lambda_{w_2},\mathcal{Q}_U^{w_2})_{w_2}}{\inf}\sum_{w_2}\lambda_{w_2}\Psi_e(\mathcal{Q}_U^{w_2}).
\end{align} where $\Psi_e(q)$ is given by Definition \ref{psieq}.

\theorem{theorem 3.1 in \cite{jet}}

\label{mainresultr2r10}
\begin{align} 
\forall \ \varepsilon>0, \ \exists \hat{n} \in \mathbb{N}, \  \forall n \geq \hat{n}, \ \ D^n_e(0,R_2) &\leq D_e^{\star}(0,R_2) + \varepsilon. \hspace{1cm}\mathrm{(achievability) } \\
 \hspace{3.13cm} \forall n \in \mathbb{N}, \ \  D_e^n(0,R_2) &\geq D_e^{\star}(0,R_2). \hspace{1.58cm}\mathrm{(converse) } 
\end{align}

\subsubsection{$(R_1,R_2) \in ]0,+\infty[^2$}
Fix a conditional probability distribution $\mathcal{Q}_{W_1,W_2|U}$. There exists $\eta>0$ such that \begin{align}
    R_2 =& I(U;W_2) + \eta, \label{eqwwe}\\
    R_1 =& I(U;W_1|W_2) + \eta. \label{eqwwee}
\end{align} 

Codebook generation: Randomly and independently generate $2^{\lfloor nR_{2}\rfloor}$ sequences $w_2^n(m_{2})$ for $m_{2} \in [1:2^{\lfloor{nR_{2}}\rfloor}]$, according to the i.i.d distribution $\mathcal{P}_{W_2^n}=\Pi_{t=1}^n\mathcal{P}_{W_2}(w_{2t})$. For each $(m_{1},m_{2}) \in [1:2^{\lfloor nR_{1} \rfloor}] \times [1:2^{\lfloor nR_{2} \rfloor}]$ generate a sequence 
$w_1^n(m_{1},m_{2})$ randomly and conditionally independently according to the i.i.d conditional distribution  $\mathcal{P}_{W_1^n|M_1W_2^n}=\Pi_{t=1}^n\mathcal{P}_{W_1|M_1W_2}(w_{1t}|m_1,w_{2t}(m_{2}))$. 
\\  
Coding algorithm: Encoder $\mathcal{E}$ observes $u^n$ and looks in the codebook for a pair $(m_{1},m_{2})$ such that \\ $(u^n, w_1^n(m_{1},m_{2}),w^n_2(m_{2})) \in \mathcal{T}_{\delta}^n(\mathcal{P}_U\mathcal{P}_{W_1W_2|U})$. If such a jointly typical tuple doesn't exist, the source encoder sets $(m_{1},m_{2})$ to $(1,1)$.
Then, it sends  $m_{2}$ to decoder $\mathcal{D}_{2}$, and $(m_{1},m_{2})$ to decoder $\mathcal{D}_{1}$.
Decoder $\mathcal{D}_{2}$ declares $w_2^n(m_{2})$ and 
decoder $\mathcal{D}_{1}$ declares $w_1^n(m_{1},m_{2})$.\\

Consider two auxiliary decoding functions $g_1$ and $g_2$ given as follows:
\begin{align}
        g_1:& \{1,2,..2^{\lfloor nR_1 \rfloor}\}\times\{1,2,..2^{\lfloor nR_2 \rfloor}\} \longrightarrow \mathcal{W}_1^n
        \\ g_2:&\{1,2,..2^{\lfloor nR_2 \rfloor}\} \longrightarrow \mathcal{W}_2^n
    \end{align}
  We assume that decoder $\mathcal{D}_{1}$ applies both decoding functions $g_1$ and $g_2$ in order to declare $(W_1^n,W_2^n)$ i.e for $M_1,M_2 \in \{1,2,..2^{\lfloor nR_1 \rfloor}\}\times\{1,2,..2^{\lfloor nR_2 \rfloor}\}$,  $\tau_1(M_1,M_2)=(g_1(M_1,M_2),g_2(M_2))$. However, for $M_2 \in \{1,2,..2^{\lfloor nR_2 \rfloor}\}$, decoder $\mathcal{D}_{2}$'s strategy $\tau_2(M_2)=g_2(M_2) \in \mathcal{W}_2^n$.

Error Event: The error event is given by $\mathcal{E}=\{(U^n,W^n_2(m_{2}),W_1^n(m_{2},m_1) \notin \mathcal{T}_{\delta}^n \}$. We have by the union of events bound $\mathcal{P}(\mathcal{E})\leq \mathcal{P}(\mathcal{E}_{1}) + \mathcal{P}(\mathcal{E}_{2}(M_2) \cap \mathcal{E}_{1}^c)$, where \begin{align}
    \mathcal{E}_{1} =& \{(U^n,W_2^n(m_{2})) \notin \mathcal{T}_{\delta}^n \  \forall m_{2}\} \\
     \mathcal{E}_{2}(m_2) =& \{(U^n,W^n_2(m_{2}),W_1^n(m_{2},m_1)) \notin \mathcal{T}_{\delta}^n \  \forall m_{1}\} 
\end{align}
 By the covering lemma,  $\mathcal{P}(\mathcal{E}_{1})$ tends to zero as $n \longrightarrow \infty$ if \begin{align}
        R_{2} >& \  I(U;W_2) + \eta. \label{ee1}
    \end{align}
     $\mathcal{P}(\mathcal{E}_{1}^c \cap \mathcal{E}_{2}(M_2))$ goes to zero by the covering lemma if \begin{align}
           R_1 + R_{2} > \ I(U;W_1,W_2)  +\eta. \label{ee2}
       \end{align} 
       
       The expected probability of error over the codebook being small means that for all $\varepsilon_2 >0$, for all $\eta>0$, there exists $\Bar{\delta}>0$, for all $\delta\leq\Bar{\delta}$, there exists $\Bar{n} \in \mathbb{N}$ such that for all $n \geq \Bar{n}$ we have: \begin{align}
           &\mathbb{E}\big[\mathcal{P}(\mathcal{E}_{1})\big] \leq \varepsilon_2, \label{eqarr} \\
           &\mathbb{E}\big[\mathcal{P}(\mathcal{E}_{2}(m_2))\big] \leq \varepsilon_2.  \label{eqarrr}
       \end{align} 

\subsubsection{Control of Beliefs} 
We introduce the indicator of error events $E^1_{\delta} \in \{0,1\}$ for decoder $\mathcal{D}_1$, and  $E^2_{\delta} \in \{0,1\}$ for decoder $\mathcal{D}_2$ defined as follows
\begin{align}
  E^1_{\delta}=&\begin{cases}
    1, & \text{if $(u^n,w_1^n,w_2^n) \notin \mathcal{T}_{\delta}^n(\mathcal{P}_{U}\mathcal{Q}_{W_1W_2|U}) $}.\\
    0, & \text{otherwise}.
  \end{cases} \\
  E^2_{\delta}=&\begin{cases}
    1, & \text{if $(u^n,w_2^n) \notin \mathcal{T}_{\delta}^n(\mathcal{P}_{U}\mathcal{Q}_{W_2|U}) $}.\\
    0, & \text{otherwise}.
  \end{cases}
\end{align}
\remark Note that $E_{\delta}^1=0 \iff (u^n,w_1^n,w_2^n) \in \mathcal{T}_{\delta}^n(\mathcal{P}_{U}\mathcal{Q}_{W_1W_2|U})  \implies (u^n,w_2^n) \in \mathcal{T}_{\delta}^n(\mathcal{P}_{U}\mathcal{Q}_{W_2|U}) \iff E^2_{\delta}=0$. Conversely, $E_{\delta}^2=1 \iff (u^n,w_2^n) \notin \mathcal{T}_{\delta}^n(\mathcal{P}_{U}\mathcal{Q}_{W_2|U})  \implies    (u^n,w_1^n,w_2^n) \notin \mathcal{T}_{\delta}^n(\mathcal{P}_{U}\mathcal{Q}_{W_1W_2|U}) \iff E_{\delta}^1=1$ Moreover, $\mathcal{P}(E_{\delta}^1=0)\leq \mathcal{P}(E_{\delta}^2=0)$ and $\mathcal{P}(E_{\delta}^1=1)\geq \mathcal{P}(E_{\delta}^2=1)$
Assuming the distribution $\mathcal{P}_{U|W_1W_2}$ is fully supported, the beliefs of decoder $\mathcal{D}_{1}$ are controlled as follows
  \begin{align}
       &\mathbb{E} \Big[\frac{1}{n}\sum_{t=1}^n D(\mathcal{P}^{m_1,m_2}_{t} ||\mathcal{P}_{U|W_{1}W_{2}}(\cdot | W_{1t},W_{2t})) \Big|E^1_{\delta}=0\Big] \\
    =& \sum_{m_1,m_2,w_1^n,w_2^n}\mathcal{P}^{\sigma,\tau_1,\tau_2}(m_1,m_2,w_1^n,w_2^n\Big| E^1_{\delta}=0)\cdot \frac{1}{n}\sum_{t=1}^n D(\mathcal{P}^{m_1,m_2}_{t} ||\mathcal{P}_{U|W_{1}W_{2}}(\cdot | W_{1t},W_{2t})) \label{b2}\\ 
     =& \sum_{m_1,m_2,w_1^n,w_2^n}\mathcal{P}^{\sigma,\tau_1,\tau_2}(m_1,m_2,w_1^n,w_2^n\Big| E^1_{\delta}=0)\cdot \frac{1}{n}\sum_{t=1}^n\sum_{u}\mathcal{P}_{t}^{m_1m_2}(u)\log_2\frac{\mathcal{P}_{t}^{m_1m_2}(u)}{\mathcal{P}_{U|W_{1}W_{2}}(u| w_{1t},w_{2t})} \label{zbb2}\\
     =& \sum_{m_1,m_2,w_1^n,w_2^n}\mathcal{P}^{\sigma,\tau_1,\tau_2}(m_1,m_2,w_1^n,w_2^n\Big| E^1_{\delta}=0)\cdot \frac{1}{n}\sum_{t=1}^n\sum_{u}\mathcal{P}_{t}^{m_1m_2}(u)\log_2 \frac{1}{\mathcal{P}_{U|W_{1}W_{2}}(u| w_{1t},w_{2t})} \nonumber \\
     &\hspace{1cm} - \sum_{m_1,m_2,w_1^n,w_2^n}\mathcal{P}^{\sigma,\tau_1,\tau_2}(m_1,m_2,w_1^n,w_2^n\Big| E^1_{\delta}=0)\cdot \frac{1}{n}\sum_{t=1}^n\sum_{u}\mathcal{P}_{t}^{m_1m_2}(u)\log_2\frac{1}{\mathcal{P}_{t}^{m_1m_2}(u)} \label{zbbb2}\\
    =& \frac{1}{n}  \sum_{m_1,m_2,w_1^n,w_2^n}\mathcal{P}^{\sigma,\tau_1,\tau_2}(m_1,m_2,w_1^n,w_2^n\Big| E^1_{\delta}=0)\cdot \frac{1}{n}\sum_{t=1}^n\sum_{u}\mathcal{P}_{t}^{m_1m_2}(u)\log_2\frac{1}{\mathcal{P}_{U|W_{1}W_{2}}(u| w_{1t},w_{2t})} \nonumber \\ &\hspace{1cm} -\frac{1}{n}\sum_{t=1}^n H(U_t|M_1,M_2, E^1_{\delta}=0) \label{b3} \\
    =& \frac{1}{n}\sum_{u^n,w_1^n,w_2^n}\mathcal{P}^{\sigma,\tau_1,\tau_2}(u^n,w_1^n,w_2^n\Big| E^1_{\delta}=0)\cdot \log_2\frac{1}{\Pi_{t=1}^n\mathcal{P}_{U|W_{1}W_{2}}(u_t| w_{1t},w_{2t})} -\frac{1}{n}\sum_{t=1}^n H(U_t|M_1,M_2, E^1_{\delta}=0) \label{b4}\\
    =& \frac{1}{n}\sum_{u^n,w_1^n,w_2^n \in \mathcal{T}_{\delta}^n}\mathcal{P}^{\sigma,\tau_1,\tau_2}(u^n,w_1^n,w_2^n\Big| E^1_{\delta}=0)\cdot \log_2\frac{1}{\Pi_{t=1}^n\mathcal{P}_{U|W_{1}W_{2}}(u_t| w_{1t},w_{2t})} -\frac{1}{n}\sum_{t=1}^n H(U_t|M_1,M_2, E^1_{\delta}=0) \label{b4444}\\
      \leq& \frac{1}{n}\sum_{u^n,w_1^n,w_2^n\in \mathcal{T}_{\delta}^n}\mathcal{P}^{\sigma,\tau_1,\tau_2}(u^n,w_1^n,w_2^n\Big| E^1_{\delta}=0)\cdot n \cdot \big(H(U|W_1,W_2) + \delta \big) -\frac{1}{n}H(U^n|M_1,M_2, E^1_{\delta}=0)  \label{b5}\\
    \leq& \frac{1}{n}I(U^n;M_1,M_2\Big| E^1_{\delta}=0)-I(U;W_1,W_2)+\delta +\frac{1}{n}+\log_2|\mathcal{U}|\cdot\mathcal{P}^{\sigma,\tau_1,\tau_2}(E^1_{\delta}=1) \label{b6}\\
   \leq& \eta + \delta +\frac{1}{n}+\log_2|\mathcal{U}|\cdot\mathcal{P}^{\sigma,\tau_1,\tau_2}(E^1_{\delta}=1)\label{b7}. 
\end{align}
\begin{itemize}
    \item Equation \eqref{b2} comes from the definition of expected K-L divergence. 
     \item Equation \eqref{zbb2} comes from the definition of K-L divergence. 
     \item Equation \eqref{zbbb2} comes from splitting the logarithm.  
    \item Equation \eqref{b3} follows since: 
    \begin{align}
     &\sum_{m_1,m_2,w_1^n,w_2^n}\mathcal{P}^{f,g_1,g_2}(m_1,m_2,w_1^n,w_2^n\Big| E^1_{\delta}=0)\cdot \frac{1}{n}\sum_{t=1}^n\sum_{u}\mathcal{P}_{t}^{m_1m_2}(u)\log_2\frac{1}{\mathcal{P}_{t}^{m_1m_2}(u)} \label{bb2} \\
     =& \sum_{m_1,m_2,w_1^n,w_2^n}\mathcal{P}^{f,g_1,g_2}(m_1,m_2,w_1^n,w_2^n\Big| E^1_{\delta}=0)\cdot \frac{1}{n}\sum_{t=1}^n H(U_t|M_1=m_1,M_2=m_2) \\ 
     =&\frac{1}{n}\sum_{t=1}^n\sum_{m_1,m_2,w_1^n,w_2^n}\mathcal{P}^{f,g_1,g_2}(m_1,m_2,w_1^n,w_2^n\Big| E^1_{\delta}=0)\cdot \ H(U_t|M_1=m_1,M_2=m_2) \\ 
     =&\frac{1}{n}\sum_{t=1}^n\sum_{m_1,m_2}\mathcal{P}^{f,g_1,g_2}(m_1,m_2\Big| E^1_{\delta}=0)\cdot \ H(U_t|M_1=m_1,M_2=m_2) \\ 
     =& \frac{1}{n}\sum_{t=1}^n H(U_t|M_1,M_2,E^1_{\delta}=0).
\end{align}

    \item Equation \eqref{b4} follows since: \begin{align}
    &\sum_{m_1,m_2,w_1^n,w_2^n}\mathcal{P}^{\sigma,\tau_1,\tau_2}(m_1,m_2,w_1^n,w_2^n\Big|E^1_{\delta}=0)\cdot \frac{1}{n}\sum_{t=1}^n\sum_{u}\mathcal{P}_{t}^{m_1m_2}(u)\log_2\frac{1}{\mathcal{P}_{U|W_{1}W_{2}}(u| w_{1t},w_{2t})} \\
    =&\frac{1}{n}\sum_{t=1}^n\sum_{u_t,m_1,m_2,w_1^n,w_2^n}\mathcal{P}^{\sigma,\tau_1,\tau_2}(u_t,m_1,m_2,w_1^n,w_2^n\Big|E^1_{\delta}=0)\cdot \log_2\frac{1}{\mathcal{P}_{U|W_{1}W_{2}}(u_t| w_{1t},w_{2t})} \\
        =&\frac{1}{n}\sum_{t=1}^n\sum_{u^n,m_1,m_2,w_1^n,w_2^n}\mathcal{P}^{\sigma,\tau_1,\tau_2}(u^n,m_1,m_2,w_1^n,w_2^n\Big|E^1_{\delta}=0)\cdot \log_2\frac{1}{\mathcal{P}_{U|W_{1}W_{2}}(u_t| w_{1t},w_{2t})} \\
        =&\frac{1}{n}\sum_{u^n,m_1,m_2,w_1^n,w_2^n}\mathcal{P}^{\sigma,\tau_1,\tau_2}(u^n,m_1,m_2,w_1^n,w_2^n\Big|E^1_{\delta}=0)\cdot \log_2\frac{1}{\Pi_{t=1}^n\mathcal{P}_{U|W_{1}W_{2}}(u_t| w_{1t},w_{2t})} \\
        =&\frac{1}{n}\sum_{u^n,w_1^n,w_2^n}\mathcal{P}^{\sigma,\tau_1,\tau_2}(u^n,w_1^n,w_2^n\Big|E^1_{\delta}=0)\cdot \log_2\frac{1}{\Pi_{t=1}^n\mathcal{P}_{U|W_{1}W_{2}}(u_t| w_{1t},w_{2t})}.
\end{align}

\item Equation \eqref{b4444} follows since the support of $\mathcal{P}^{\sigma,\tau_1,\tau_2}(u^n,w_1^n,w_2^n|E^1_{\delta})=\mathbb{P}\{(u^n,w_1^n,w_2^n) \in \mathcal{T}_{\delta}^n \}$ is included in $\mathcal{T}_{\delta}^n$. 
    \item Equation \eqref{b5} follows from the typical average lemma property (Property 1 pp.26 in \cite{elgamal}) given in lemma \ref{a.20jet}, and the chain rule of entropy: $H(U^n|M_1,M_2,W^n_1,W^n_2) \leq \sum_{t=1}^n H(U_t|M_1,M_2,W_1,W_2)$. 
   \item Equation \eqref{b6} comes from the conditional entropy property and the fact that $H(U^n)=nH(U)$ for an i.i.d random variable $U$ and lemma \ref{a.2222jet}.
    \item Equation \eqref{b7} follows since $I(U^n;M_1,M_2) \leq H(M_1,M_2) \leq \log_2 |J| = n \cdot (R_1+R_2) = n \cdot (I(U;W_1,W_2)+\eta)$ and lemma \ref{a.2222jet}.
\end{itemize}

Similarly for decoder $\mathcal{D}_{2}$ we have

  \begin{align}
       &\mathbb{E} \Big[\frac{1}{n}\sum_{t=1}^n D(\mathcal{P}^{m_2}_{t} ||\Pi_{t=1}^n\mathcal{P}_{U_t|W_{2t}})\Big|E^2_{\delta}=0\Big] \\
    =& \sum_{m_2,w_2^n}\mathcal{P}^{\sigma,\tau_2}(m_2,w_2^n\Big|E^2_{\delta}=0)\cdot \frac{1}{n}\sum_{t=1}^n D(\mathcal{P}^{m_2}_{t} ||\Pi_{t=1}^n\mathcal{P}_{U_t|W_{2t}}) \label{x2b2}\\ 
    =& \frac{1}{n}  \sum_{(u^n,m_2,w^n_2) \in \mathcal{T}^n_{\delta}}\mathcal{P}^{\sigma,\tau_2}(m_2,w_2^n\Big|E^2_{\delta}=0)\cdot \log_2\frac{1}{\Pi_{t=1}^n\mathcal{P}_{U_t|W_{2t}}} -\frac{1}{n}\sum_{t=1}^n H(U_t|M_2,E^2_{\delta}=0) \label{x2b3} \\
    \leq& \frac{1}{n}  \sum_{(u^n,m_2,w^n_2) \in \mathcal{T}^n_{\delta}}\mathcal{P}^{\sigma,\tau_2}(u^n,m_2,w^n_2\Big|E^2_{\delta}=0)\cdot \log_2\frac{1}{\Pi_{t=1}^n\mathcal{P}_{U_t|W_{2t}}} -\frac{1}{n}\sum_{t=1}^n H(U_t|M_2,E^2_{\delta}=0) \label{x2b4}\\
      \leq& \frac{1}{n}  \sum_{(u^n,m_2,w^n_2) \in \mathcal{T}^n_{\delta}}\mathcal{P}^{\sigma,\tau_2}(u^n,m_2,w^n_2 \Big|E^2_{\delta}=0)\cdot n \cdot \big(H(U|W_2) + \delta \big) -\frac{1}{n}H(U^n|M_2,E^2_{\delta}=0) \label{x2b5}\\
    \leq& \frac{1}{n}I(U^n;M_2,E^2_{\delta}=0)-I(U;W_2)+\delta +\frac{1}{n}+\log_2|\mathcal{U}|\cdot\mathcal{P}^{\sigma,\tau_2}(E^2_{\delta}=1) \label{x2b6}\\
    \leq& \eta + \delta +\frac{1}{n}+\log_2|\mathcal{U}|\cdot\mathcal{P}^{\sigma,\tau_2}(E^2_{\delta}=1) \label{x2b8} \\
    \leq& \eta + \delta +\frac{1}{n}+\log_2|\mathcal{U}|\cdot\mathcal{P}^{\sigma,\tau_1,\tau_2}(E_{\delta}^1=1).
\end{align}


If the expected probability of error is small over the codebooks, then it has to be small over at least one codebook. Therefore, equations  \eqref{eqarr} and \eqref{eqarrr} imply that: \begin{align}
    \forall \epsilon_2>0, \forall \eta>0, \exists \Bar{\delta}>0,\forall \delta \leq \Bar{\delta}, \exists \Bar{n} \in \mathbb{N}, \forall n\geq \Bar{n}, \exists b^{\star}, \ \mathrm{s.t.} \ \mathcal{P}_{b^{\star}}(E^2_{\delta}=1)\leq \varepsilon_2. \label{eqwaa} 
\end{align}     
       The strategy $\sigma$ of the encoder consists of using $b^{\star}$ in order to transmit the pair $(m_1,m_2)$ such that $(U^n,W^n_2(m_{2}),W_1^n(m_{2},m_1)$ is a jointly typical sequence. By construction, this satisfies equation \eqref{eqwaa}. 
       
\lemma \label{lemmmaa10}Let $\mathcal{Q}_{W_1W_2|U} \in \tilde{Q}_{0}(R_1,R_2)$, then $\forall \varepsilon>0$, $\forall \alpha>0,\gamma>0$, there exists $\Bar{\delta}$, $\forall \delta \leq\Bar{\delta}$, $\exists \Bar{n}$, $\forall n \geq \Bar{n}$, $\exists \sigma$, such that   $1-\mathcal{P}^{\sigma}(B_{\alpha,\gamma,\delta})\leq \varepsilon$. \\
      \proof{of lemma \ref{lemmmaa10}}
      We have:  \begin{align}
          &1-\mathcal{P}_{\sigma}(B_{\alpha,\gamma,\delta}):=\mathcal{P}_{\sigma}(B^c_{\alpha,\gamma,\delta}) \\
          &= \mathcal{P}_{\sigma}(E^2_{\delta}=1)\mathcal{P}_{\sigma}(B^c_{\alpha,\gamma,\delta}|E^2_{\delta}=1) +  \mathcal{P}_{\sigma}(E^1_{\delta}=0)\mathcal{P}_{\sigma}(B^c_{\alpha,\gamma,\delta}|E^1_{\delta}=0) \\
          &\leq \mathcal{P}_{\sigma}(E^2_{\delta}=1) + \mathcal{P}_{\sigma}(B^c_{\alpha,\gamma,\delta}|E^1_{\delta}=0) \\
          &\leq \varepsilon_{2} + \mathcal{P}_{\sigma}(B^c_{\alpha,\gamma,\delta}|E^2_{\delta}=1). \label{eqw222}
      \end{align}
      Moreover, \begin{align}
          \mathcal{P}_{\sigma}(B^c_{\alpha,\gamma,\delta}|E^1_{\delta}=0) &= \sum_{w_1^n,w_2^n,m_1,m_2}\mathcal{P}_{\sigma}\bigg((w_1^n,w_2^n,m_1,m_2)\in B^c_{\alpha,\gamma,\delta}\Bigg|E^1_{\delta}=0\bigg) \label{eqa1}\\
          &= \sum_{w_1^n,w_2^n,m_1,m_2}\mathcal{P}_{\sigma}\bigg((w_1^n,w_2^n,m_1,m_2) \ \mathrm{ s.t. } \ \frac{|T_{\alpha}(w_1^n,w_2^n,m_1,m_2)|}{n}\leq 1-\gamma \Bigg|E^1_{\delta}=0\bigg) \\
          &= \mathcal{P}_{\sigma}\bigg(\frac{\#}{n}\bigg\{t, D\bigg(\mathcal{P}^{m_1,m_2}_{t} \bigg|\bigg|\mathcal{Q}_{U|W_{1}W_{2}}(\cdot | W_{1t},W_{2t})\bigg) \leq \frac{\alpha^2}{2\ln2}<1-\gamma\Bigg| E^1_{\delta}=0\bigg\} \\
          &= \mathcal{P}_{\sigma}\bigg(\frac{\#}{n}\bigg\{t, D\bigg(\mathcal{P}^{m_1,m_2}_{t} \bigg|\bigg|\mathcal{Q}_{U|W_{1}W_{2}}(\cdot | W_{1t},W_{2t})\bigg) > \frac{\alpha^2}{2\ln2} \geq \gamma\Bigg| E^1_{\delta}=0\bigg\} \label{eqa4}\\
          &\leq \frac{2\ln{2}}{\alpha^2\gamma}\cdot\mathbb{E}_{\sigma}\bigg[\frac{1}{n}\sum_{t=1}^n D\bigg(\mathcal{P}^{m_1,m_2}_{t} \bigg|\bigg|\mathcal{Q}_{U|W_{1}W_{2}}(\cdot | W_{1t},W_{2t})\bigg)\bigg] \label{eqa5}\\
          &\leq \frac{2\ln{2}}{\alpha^2\gamma}\cdot \bigg(\eta +\delta+\frac{2}{n}+2\log_2|\mathcal{U}|\cdot\mathcal{P}_{\sigma}(E^2_{\delta}=1)\bigg) \label{eqa6}
      \end{align}
      \begin{itemize}
          \item Equations \eqref{eqa1} to \eqref{eqa4} are simple reformulations. 
          \item Equation \eqref{eqa5} comes from using Markov's inequality given in lemma \ref{eqa0}.  
          \item Equation \eqref{eqa6} comes from equations \eqref{b7} and \eqref{x2b8}.
      \end{itemize}
\lemma (Markov's Inequality)\label{eqa0}. For all $\varepsilon_1>0$ , $\varepsilon_2>0$ we have:
\begin{align}
    &\mathbb{E}_{\sigma}\bigg[\frac{1}{n}\sum_{t=1}^n D\bigg(\mathcal{P}^{m_1,m_2}_{t} \bigg|\bigg|\mathcal{Q}_{U|W_{1}W_{2}}(\cdot | W_{1t},W_{2t})\bigg)\bigg] \leq \varepsilon_0 \\
    &\implies \mathcal{P}_{w_1^n,w_2^n,m_1,m_2}\bigg(\frac{\#}{n}\bigg\{t, D\bigg(\mathcal{P}^{m_1,m_2}_{t} \bigg|\bigg|\mathcal{Q}_{U|W_{1}W_{2}}(\cdot | W_{1t},W_{2t})\bigg)>\varepsilon_1\bigg\}>\varepsilon_2\bigg)\leq \frac{\varepsilon_0}{\varepsilon_1\cdot\varepsilon_2}.
\end{align} 
      \proof{of lemma \ref{eqa0}} We denote by $D_t=D(\mathcal{P}^{m_1,m_2}_{t} ||\mathcal{Q}_{U|W_{1}W_{2}}(\cdot | W_{1t},W_{2t})$ and $D^n=\{D_t\}_t$ the K-L divergence. We have that:
      \begin{align}
          \mathcal{P}\bigg(\frac{\#}{n}\bigg\{t, \mathrm{s.t.} D_t>\varepsilon_1\bigg\}>\varepsilon_2\bigg) =&\mathcal{P}\bigg(\frac{1}{n}\cdot\sum^n_{t=1}\mathbbm{1}\bigg\{D_t>\varepsilon_1\bigg\}>\varepsilon_2 \bigg) \label{markovchain1}\\
          \leq& \frac{\mathbb{E}\bigg[\frac{1}{n}\cdot\sum_{t=1}^n\mathbbm{1}\bigg\{D_t>\varepsilon_1\bigg\}\bigg]}{\varepsilon_2} \label{markovchain2}\\
          =& \frac{\frac{1}{n}\sum_{t=1}^n\mathbb{E}\bigg[\mathbbm{1}\bigg\{D_t>\varepsilon_1\bigg\}\bigg]}{\varepsilon_2} \label{markovchain3}\\
          =& \frac{\frac{1}{n}\sum_{t=1}^n\mathcal{P}\bigg(D_t>\varepsilon_1\bigg)}{\varepsilon_2} \label{markovchain4} \\
          \leq& \frac{\frac{1}{n}\sum_{t=1}^n\frac{\mathbb{E}\bigg[D_t\bigg]}{\varepsilon_1}}{\varepsilon_2} \label{markovchain5}\\
          =& \frac{1}{\varepsilon_1\cdot\varepsilon_2}\cdot \mathbb{E}\bigg[\frac{1}{n}\sum^n_{t=1}D_t\bigg] \leq \frac{\varepsilon_0}{\varepsilon_1\cdot\varepsilon_2}. \label{markovchain6}
      \end{align}
       \begin{itemize}
           \item Equations \eqref{markovchain1},  \eqref{markovchain3}, \eqref{markovchain4} and \eqref{markovchain6} are reformulations of probabilities and expectations.
           \item Equations \eqref{markovchain2} and \eqref{markovchain5}, come from Markov's inequality $\mathcal{P}(X\geq \alpha)\leq\frac{\mathbb{E}[X]}{\alpha}, \ \forall \alpha>0$.
       \end{itemize}

Combining equations  \eqref{eqwaa}, \eqref{eqw222}, and \eqref{markovchain6} we get the following statement:
\begin{align}
    \forall \epsilon_3>0, \forall \alpha>0, \forall \gamma >0, \exists \Bar{\eta}, \forall\eta\leq \Bar{\eta}, \exists \Bar{\delta}>0,\forall \delta \leq \Bar{\delta}, \exists \Bar{n} \in \mathbb{N}, \forall n\geq \Bar{n}, \exists \sigma \end{align} such that 
    \begin{align}
      \mathcal{P}_{\sigma}(B^c_{\alpha,\gamma,\delta}) \leq 2\cdot\mathcal{P}_{\sigma}(E^2_{\delta}=1) + \frac{2\ln2}{\alpha^2\gamma}\cdot\bigg(\eta+\delta+\frac{2}{n}+2\log_2|\mathcal{U}|\cdot\mathcal{P}_{\sigma}(E^2_{\delta}=1)\bigg)\leq \varepsilon_3. 
       \end{align}
       By choosing appropriately the rates $(R_1,R_2)$ in \eqref{eqwwe} and \eqref{eqwwee} such as to make $\eta >0$ small, we obtain the desired result: 
       \begin{align}
           \forall \varepsilon>0, \forall \alpha>0,\gamma>0, \exists \Bar{\delta}, \forall \delta \leq \Bar{\delta}, \exists \Bar{n}, \forall n \geq \Bar{n}, \exists \sigma, \ \mathrm{s.t}  \   1-\mathcal{P}^{\sigma}(B_{\alpha,\gamma,\delta})\leq \varepsilon.
       \end{align}
This completes the proof of achievability.
\endproof{}

\section{More Lemmas}

\lemma{(Typical Sequences Property 1, pp.26 in \cite{elgamal})}.
The typical  sequences $(u^n,w_1^n,w_2^n) \in \mathcal{T}_{\delta}^n$ satisfy: \begin{align}
    \forall \varepsilon>0, \  \exists \Bar{\delta}>0, \  \forall \delta \leq \Bar{\delta}, \  \forall n, \  \forall (u^n,w_1^n,w_2^n) \in \mathcal{T}_{\delta}^n, \nonumber \\ \Bigg|\frac{1}{n}\cdot\log_2\frac{1}{\Pi_{t=1}^n\mathcal{P}(u|w_{1t},w_{2t})}-H(U|W_1,W_2) \Bigg|\leq \varepsilon, \end{align}   where  $\Bar{\delta}=\varepsilon\cdot H(U|W_1,W_2)$.

\label{a.20jet} 

\lemma \label{a.2222jet} Let $U^n$ an i.i.d random variable and $M$ a random variable. For all $\varepsilon>0$, there exists $\Bar{n}\in \mathbb{N}$, such that for all $n \geq \Bar{n}$, we have  \begin{align}
    H(U^n|E_{\delta}=0) \geq n\cdot\Big(H(U)-\varepsilon).
\end{align}
\begin{proof} 
\begin{align}
    H(U^n|E_{\delta}=0)=&\frac{1}{\mathcal{P}(E_{\delta}=0)}\cdot\Big(H(U^n|E_{\delta}=1)-\mathcal{P}(E_{\delta}=1)\cdot H(U^n|E_{\delta}=1)\Big) \label{rq1} \\
    \geq& H(U^n|E_{\delta})-\mathcal{P}(E_{\delta}=1)\cdot H(U^n|E_{\delta}=1)\Big) \label{rq2} \\
    \geq& H(U^n)-H(E_{\delta})-\mathcal{P}(E_{\delta}=1)\cdot H(U^n|E_{\delta}=1)\Big) \label{rq3} \\
    \geq& H(U^n)-n\cdot \varepsilon \label{rq4} .
\end{align} 
\end{proof}

\begin{itemize}
    \item Equation \eqref{rq1} follows from the conditional entropy definition.
    \item Equation \eqref{rq2} follows since $\mathcal{P}(E_{\delta}=0) \leq 1$.
    \item Equation \eqref{rq3} comes from the property $H(U^n|M,E_{\delta})=H(U^n,M,E_{\delta})-H(M) - H(E_{\delta}) \geq H(U^n) - -H(M) - H(E_{\delta})$. 
    \item Equation \eqref{rq4} follows since $U$ is i.i.d and the definition of $E_{\delta}=1$.Hence, for all $\varepsilon$, there exists an $\Bar{n}\in \mathbb{N}$ such that for all $n \geq \Bar{n}$ we have $H(\mathcal{P}(E_{\delta}=1))+H(M)+\mathcal{P}(E_{\delta}=1)\cdot \log_2|\mathcal{U}| \leq \varepsilon.$ 
\end{itemize}

\section*{Acknowledgment}

The authors thank Tristan Tomala for fruitful discussions regarding the equation \eqref{eq:T} in 
the converse proof.

\bibliographystyle{IEEEtran}
\bibliography{sample}

\end{document}